\RequirePackage{fix-cm}
\documentclass[smallextended]{svjour3}       
\smartqed  
\usepackage{graphicx}
\usepackage{algorithm}
\usepackage{algpseudocode}
\usepackage{amsmath}
\usepackage{amsfonts}
\usepackage{enumerate}
\usepackage{url}
\usepackage[labelfont=bf]{caption}  
\usepackage{subfloat}

\usepackage[authoryear]{natbib}
\setcitestyle{aysep={}}

\graphicspath{ {figures/} }

\newcommand{\R}{\mathfrak{R}}
\newcommand{\D}{\mathfrak{D}}
\renewcommand{\L}{\mathfrak{L}}

 \journalname{ }

\begin{document}

\title{Gene tree species tree reconciliation with gene conversion \thanks{This work is funded by the Agence Nationale pour la Recherche, Ancestrome project ANR-10-BINF-01-01.}
}


\author{Damir Hasi\'{c}         \and
        Eric Tannier 
}


\institute{Damir Hasi\'{c} \at
              Department of Mathematics, Faculty of Science, University of Sarajevo, 71000 Sarajevo, Bosnia and Herzegovina \\
              \email{damir.hasic@gmail.com, d.hasic@pmf.unsa.ba}           
           \and
           Eric Tannier \at
              Inria Grenoble Rh\^one-Alpes, F-38334 Montbonnot, France\\
              Univ Lyon, Universit\'e Lyon 1, CNRS, Laboratoire de Biom\'etrie et Biologie \'Evolutive UMR5558, F-69622 Villeurbanne, France             
}

\date{ }

\maketitle
    \begin{abstract}
    	Gene tree/species tree reconciliation is a recent decisive progress in phylogenetic methods, accounting for the possible differences between gene histories and species histories. Reconciliation consists in explaining these differences by gene-scale events such as duplication, loss, transfer, which translates mathematically into a mapping between gene tree nodes and species tree nodes or branches.
Gene conversion is a frequent and important biological event, which results in the replacement of a gene by a copy of another from the same species and in the same gene tree. Including this event in reconciliations has never been attempted because this changes as well the solutions as the methods to construct reconciliations. Standard algorithms based on dynamic programming become ineffective. 
We propose here a novel mathematical framework including gene conversion as an evolutionary event in gene tree/species tree reconciliation. We describe a randomized algorithm giving in polynomial running time a reconciliation minimizing the number of duplications, losses and conversions. We show that the space of reconciliations includes an analog of the Last Common Ancestor reconciliation, but is not limited to it. Our algorithm outputs any optimal reconciliation with non null probability.
We argue that this study opens a research avenue on including gene conversion in reconciliation, which can be important for biology.

    \end{abstract}

	\section{Introduction}

\subsection{Biological motivation}

Due to various evolutionary events on a gene level, gene trees (trees used to describe the evolution of genes) and species trees (trees used to describe the evolution of species) are often not identical. Identifying these evolutionary events, such as speciation, duplication, transfer, conversion, transfer with replacement, and their positioning inside species tree is called {\em phylogenetic reconciliation}.

	Tree reconciliation techniques become widely used in biology. For example they are used in testing hypotheses of horizontal transfer in some Bacterial and Archaeal species \citep{Planet2003}; studying parasites infecting tropheine cichlids \citep{Vanhove2015}; finding horizontal gene transfers of RH50 among prokaryotes \citep{pmid28049420}. Reconciliation tools \citep{Szllsi23102012, doi:10.1093/sysbio/syt054, doi:10.1093/sysbio/syt003} are also used to explore the process of shaping gut microbiomes \citep{Groussin2017}. 
	In \cite{pmid15713731, pmid11836216, pmid17346331} reconciliations are used "for inferring orthology relationships" \citep{Doyon2011}, and in \cite{pmid15306658, Searls2003} "for identifying orthologs for use in function prediction, gene annotation, planning experiments in model organisms, and identifying drug targets" \citep{Vernot2008}. 
	From \cite{pmid21238340, JBI:JBI1315} we can see that "reconciliation can also be used to study co-evolution between parasites and their hosts (parasitology), and between organisms and their living areas (biogeography)" \citep{Doyon2011}.

An evolutionary event of particular interest in this paper is {\em gene conversion}. It is a highly important genomic event for evolution and health \citep{Chen2007}. It results in the replacement of a gene in a genome by another homologous gene from the same genome, where homologous means that they have a common ancestor. It has largely contributed to shaping extant eukaryotic genomes and is involved in several known human genetic diseases \citep{Ko2011}.

However, gene conversion is nearly absent from the mathematical framework for phylogeny. Phylogenetic methods can handle base substitutions, indels \citep{Felsenstein2004}, genome rearrangements \citep{Hu2014}, duplications, transfers and losses of genes \citep{Szllosi2015} or population scale events as incomplete lineage sorting \citep{Mirarab2014}. But the detection of gene conversion is still done with empirical examinations of gene trees combined with various genomic features \citep{Hsu2010,Mansai2010}.

This absence of gene conversion can strongly bias evolutionary studies. Indeed, it introduces a discordance between the history of a gene and the history of a locus \citep{Rasmussen2012} which stays unresolved. It makes the confusion between duplications and conversions \citep{Boussau2013a}, whereas conversions are probably more frequent \citep{Kejnovsky2007}.

\subsection{Mathematical and computational aspects of the problem}

With $V(T)$ we denote the set of all nodes, and $L(T)$ is the set of all leaves of a tree $T$. We assume that a gene tree $G$ and a species tree $S$ are given, as well as a mapping $\phi:L(G)\rightarrow L(S)$ that places extant genes into extant species. 

The problem is to find a mapping $\rho:V(G)\rightarrow V(T)$ that optimizes some objective function. How to determine $\rho$ depends on a model that describes a problem of reconciliation. The model includes the set of allowed evolutionary events (speciation is usually always included) and the objective function, which is usually the likelihood of a reconciliation (maximization problem) or the weight of a reconciliation (minimization problem). The weight of a reconciliation, which is the sum of costs of all evolutionary events in a reconciliation, is a sort of measure of dissimilarity between $G$ and $S$.         

In this paper, the objective function is the weight of a reconciliation. Conversions are modeled as a pair of duplication and loss. Since we are pairing gene losses with gene duplications, there is a need to introduce {\em lost subtrees}, {\em i.e.} subtrees of the gene tree that were not given in the input. This means that, in order to obtain an optimal solution, we need to extend given gene tree $G$, and this extension we denote by $G'$. Because of pairing losses with duplications, we obtain that disjoint subtrees of $G$ are not independent anymore. The loss of independence and the need to extend the given gene tree are things that make the problem harder than the usual duplication/loss reconciliation.

\subsection{A review of some previous results}

The first model of reconciliation to mention is the one with duplications, speciations and losses. A natural way to form a reconciliation, in this model, is to position every node from the gene tree as low as possible inside the species tree. This type of reconciliation is called the {\em Last Common Ancestor} (LCA). LCA minimizes the number of duplications and losses \citep{Gorecki2006}, the number of duplications \citep{Gorecki2006}, and the number of losses \citep{Chauve2009, Chauve2008}. LCA is the only reconciliation that minimizes duplications and losses \citep{Gorecki2006}. These reconciliations can be found in linear time. There is a polynomial algorithm in \cite{Vernot2008} that finds the minimum number of duplications even when $S$ is polytomous. 
  The problem of reconciliation between a polytomous gene tree and a binary species tree minimizing the number of mutations (duplications + losses) is polynomial \citep{Chang2006, Lafond2012}. In \cite{Zheng2017}, $O(|G|+|S|)$ algorithms for reconciling a nonbinary gene tree and a binary species tree in the duplication, loss, mutation, and deep coalescence models are given.

A biologically important and mathematically much studied evolutionary event is {\em gene transfer}. Models that include duplications, losses, and transfer are called {\em DTL models}.
When the transfers are included, then time constraints are introduced, because direct gene transfer can happen only between species that exist in the same moment. There are two ways of considering time constraints in reconciliations. One is to use an undated species tree but imposing a consistency between found transfers. This variant has been proved to be NP-hard in \cite{Hallett2004} (while without time consistency it is solvable in time $O(m^2n)$, where $m$ is the number of extant species and $n$ is the number of extant genes). Another is to use a fully dated species tree as an input, that is, there is a total order on the internal nodes. In that case a reconciliation algorithm with duplications, transfers and losses is given in \cite{Doyon2010} with time complexity $\Theta(m^2n)$. In \cite{Chan2015} the space of all reconciliations is explored and formula for its size is given. Discrete and continuous cases for DTL model are equivalent \citep{Ranwez2016}. In \cite{CHAN20171}, duplications, transfers, losses, and incomplete lineage sorting are included in the model and the FPT (fixed-parameter-tractable) algorithm for the most parsimonious reconciliation is given. If a gene that is transfered replaces another gene, then we have {\em transfer with replacement}, which is to transfer what conversion is to duplication (see \cite{TRonly2017} for NP-hardness proof, and FPT algorithm)
For a more detailed review on reconciliations see \cite{Szllosi2015}, \cite{Nakhleh2012}, and \cite{Doyon2011}.

\subsection{The contribution of this paper}

Gene conversion can be modeled in the gene tree/species tree reconciliation framework. It consists in coupling a duplication (the donor sequence) and a loss (the receiver sequence). It is usually not included in reconciliation models because the usual algorithmic toolbox of gene tree/species tree reconciliation, based on dynamic programming assuming a statistical independence between lineages, does not allow to couple events from different lineages. 

Our contribution is to explore the algorithmic possibilities of introducing conversion in reconciliations. We formally define a reconciliation with duplications, losses and conversions. We define the algorithmic problem of computing, given a gene tree and a species tree, a reconciliation minimizing a linear combination of the number of events of each type. We fully solve the problem in the particular case when all events are equally weighted. More precisely, we construct an algorithm which gives, in polynomial running time, an optimal solution, and we prove that any optimal solution can be output by the algorithm with a non null probability. The algorithm can be used as a polynomial delay enumeration of the whole space of solutions.

The space of solutions is non trivial. In contrast with the duplication and loss only reconciliations, solutions are not unique, they are not all given by the standard Last Common Ancestor (LCA) technique. Moreover, easy examples show that the LCA technique does not give the optimal solution if events are weighted differently. This opens a wide range of new open algorithmic problems related to gene tree/species tree reconciliations.

The paper is organized as follows. Section \ref{sec:definitions} introduces a gene tree/species tree reconciliation including gene conversion events, and states the relations with the classical duplication loss reconciliation. Section \ref{sec:LCA} is devoted to the presentation of an algorithm to find one optimal solution, which is called an LCA completion. In Section \ref{sec:all}, we give an algorithm to find all optimal solutions, by the definition of a class of optimal solutions called zero-flow, containing but not limited to LCA completions. We prove that an algorithm finding all zero-flow reconciliations is sufficient to access the whole solution space, and we write such an algorithm. In Section \ref{sec:algorithm} we complete the proof that the presented algorithm always gives an optimal solution, and that every optimal solution can be output with a non null probability.


\section{Reconciliations with Duplication, Loss, Conversion}\label{sec:definitions}

In this section we define the mathematical problem modeling the presence of gene conversion in gene tree species tree reconciliations. We start with the definition of the standard duplication and loss model, and then add the possibility of conversions.


\subsection{Duplication-Loss reconciliations}

Let us begin with some generalities about phylogenetic trees.
All phylogenetic trees are binary rooted trees where the root node has degree 1, and its incident edge is called the \emph{root edge}. The root edge of $T$ is denoted by $root_E(T)$, and the root node by $root(T)$. If $x$ is a node in a tree, then $L(x)$ denotes the set of leaves of the maximal subtree rooted at $x$. If $x\in V(T)\backslash L(T)$ then $x_r,x_l$ denote the two children of $x$. Similarly, we can define the children $e_r,e_l$ of an edge $e$. If $x$ is a leaf or an edge incident to a leaf, then their children are NULL and $f(NULL)=0$ for any function/procedure which returns some value.
If $x$ is a node/edge in a rooted tree $T$, then $p_T(x)=p(x)$ denotes its parent. Let $e=(x,p(x))$ be an edge, then $T(e)$ denotes the maximal rooted subtree with root edge $e$. If $x$ is on the path from $y$ to $root(T)$ then we say that $x$ is an {\em ancestor} of $y$, or that $y$ is a {\em descendant} of $x$, and we write $y\le_T x$ or $y\le x$, defining a partial order on the nodes. If $x$ is neither ancestor nor descendant of $y$, we say that $x$ and $y$ are \emph{incomparable}.
Let $x$ and $y$ be comparable nodes in a rooted tree $T$, then with $d_T(x,y)$ or $d(x,y)$ we denote the distance, {\em i.e.} the number of edges in the path between $x$ and $y$.
For a partially ordered set $A$, we use {\em minimal} to denote an element $m$ such that $x\le m \implies x=m$, $\forall x\in A$. 
We use this terminology for the partial order defined by rooted trees. For example,  if $V'$ is a subset of nodes of a tree, their {\em Last Common Ancestor} (LCA) is the minimal node which is an ancestor of all nodes in $V'$. We also use it for partial orders defined by inclusion on sets or by subtrees in trees. In particular we can use it for the partial order defined by the {\em extension} relation.

\begin{definition}[Extension]
	A tree $G'$ is said to be an {\em extension} of a gene tree $G$ if $G$ can be obtained from $G'$ by pruning some subtrees and suppressing nodes of degree 2.    
\end{definition}

We define the gene tree species tree duplication loss (DL) reconciliation.
We suppose we have two trees $G$ and $S$, respectively called the {\em gene tree} and the {\em species tree}. Nodes of $G$ ($S$) are called {\em genes} ({\em species}). A mapping $\phi:L(G) \rightarrow L(S)$ indicates the species in which genes are found in the data. Without loss of generality we suppose that $\phi$ verifies that the last common ancestor of all the leaves of $S$ that are in the image of $\phi$ is the node adjacent to the root node (recall the root node has degree 1).
The reconciliation is based on a function $\rho$, which is an extension of $\phi$ to all genes and species, including internal nodes.

\begin{definition}[Consistency]\label{def:consis}
	A function $\rho:V(G')\rightarrow V(S)$ on the nodes of a tree $G'$ is said to be {\em consistent} with a species tree S if $\rho(root(G'))=root(S)$ and for every $x\in V(G')\backslash L(G')$ one of the conditions holds (D) $\rho(x)=\rho(x_l)=\rho(x_r)$ or (S) $\rho(x)_l=\rho(x_l)$ and $\rho(x)_r=\rho(x_r)$. We also say that $G'$ is $\rho$-consistent with $S$.   
\end{definition}

Obviously, both conditions $(D)$ and $(S)$ cannot hold for a single node. 

\begin{definition}[DL reconciliation]
	Let $G$ and $S$ be a gene and a species trees and $\phi:L(G) \rightarrow L(S)$. A {\em DL reconciliation} between $G$ and $S$ is a 5-tuple $(G,G',S, \phi, \rho)$ such that $G'$ is an extension of $G$, $G'$ is $\rho$-consistent with $S$, and $\rho/L(G)=\phi$. 
\end{definition}
Note that we allow some extant species not to have genes. The definition is equivalent to the standard ones \cite{Arvestad2004, Gorecki2006, Chauve2009}, although they can present some variations between them. For example we do no impose that losses are  represented by subtrees extended to the leaves of $S$ (which is the case for example in \cite{Chauve2009}), because of the particular use we make of loss subtrees in the sequel. An example of DL reconciliation is given in Figure \ref{fig:reconciliation_examples} (a).

\begin{figure}[H]
	\centering
	\includegraphics{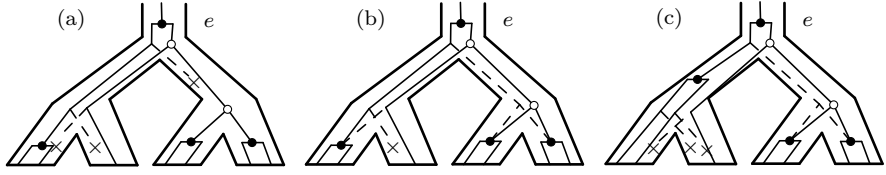}
	\caption{Examples of reconciliations. The gene tree is depicted inside the species tree to signify the mapping $\rho$. Duplication nodes are black circles, speciation nodes are white circles, losses are leaves with crosses and conversions are duplication nodes which are also a leaf of the lost subtrees, which are dashed. (a) An LCA reconciliation. Total cost: $3l+4d=7$. (b) An LCA completion, obtained from LCA by extending losses and assigning them to duplications. Total cost: $l+d+3c=5$. (c) A non-optimal reconciliation. Total cost: $3l+2d+2c=7$}
    \label{fig:reconciliation_examples}
\end{figure}

\begin{definition}[Duplication]
	Let $\R=(G,G',S,\phi,\rho)$ be a DL reconciliation and $x\in V(G')\backslash L(G')$ satisfies condition (D). Then $x$ is called a {\em duplication}. The set of all duplications is denoted by $\Delta=\Delta(\R)$.   
\end{definition}

\begin{definition}[Speciation]
	Let $\R=(G,G',S,\phi,\rho)$ be a DL reconciliation and $x\in V(G')\backslash L(G')$ satisfies condition (S). Then $x$ is called a {\em speciation}. The set of all speciations is denoted by $\Sigma=\Sigma(\R)$.   
\end{definition}

\begin{definition}[Loss]
	Let $\R=(G,G',S,\phi,\rho)$ be a DL reconciliation and $x\in L(G')\backslash L(G)$. Then $x$ is called a {\em loss}. The set of all losses is denoted by $\Lambda=\Lambda(\R)$.    
\end{definition}

We say that a duplication, loss or speciation $x$ is {\em assigned} to $s$ if $\rho(x)=s$. Let $\L(s,\R)=\L(s)=|\rho^{-1}(s)\cap \Lambda(\R)|$ and $\D(s,\R)=\D(s)=|\rho^{-1}(s)\cap \Delta(\R)|$ be the number of losses and the number of duplications assigned to $s\in V(S)$ in the reconciliation $\R$. If $e=(s,p(s))\in E(S)$, then $\L(e,\R)=\L(e)=\L(s,\R)$ and $\D(e,\R)=\D(e)=\D(s,\R)$.

The next definition extends the notion of loss. 
\begin{definition}[Lost subtree]
	Let $\R=(G,G',S,\phi,\rho)$ be a DL reconciliation. A maximal subtree $T$ of $G'$ such that $V(T)\cap V(G) =\emptyset$ is called a {\em lost subtree}.
\end{definition}

The next lemma introduces the standard Last Common Ancestor reconciliation, and its proof can be found in \cite{Chauve2009} or \cite{Chauve2008}.

\begin{lemma}\label{lem:LCA_def}
	Let $G$ and $S$ be a gene and a species tree, and $\phi:L(G)\rightarrow L(S)$. There exists a DL reconciliation $\R=(G,G',S, \phi, \rho)$ such that $\rho(x)$ is the root of the minimal subtree of $S$ containing $L(\phi(x))$, $\forall x\in V(G)$.  
\end{lemma}

\begin{definition}[LCA reconciliation]
	The DL reconciliation from Lemma $\ref{lem:LCA_def}$ that minimizes $|\Lambda(\R)|$ is called the {\em Last Common Ancestor (LCA)} reconciliation
    and is noted $\R_{lca}=(G,G'_{lca},S, \phi, \rho_{lca})$.
\end{definition}

Note that the LCA reconciliation is the unique reconciliation minimizing the number of duplications, or the number of losses, or any linear combination of these two numbers \cite{Chauve2009}. In Section \ref{sec:LCA} we will construct equivalents of the LCA reconciliation including conversions, called LCA completions, which will have the property of minimizing the sum of the number of duplications, losses and conversions. However in contrast it is not unique, it does not contain all optimal solutions (as we show it in Section \ref{sec:all}) and does not optimize over any linear combinations of these numbers (see the conclusion for such an example).


\subsection{Duplication-Loss-Conversion reconciliations}

In the next definition we introduce an additional event, called {\em gene conversion}, which is a function $\delta$ pairing some losses and duplications. This models the replacement of a gene by a copy of another one from the same family.
\begin{definition}[Conversion]
	Let $(G,G',S,\phi,\rho)$ be a DL reconciliation. Let $\delta: \Delta \rightarrow \Lambda $ be an injective partial function such that $\rho(x)=\rho(\delta(x))$ for all $x\in \delta^{-1}(\Lambda)$. If $x\in \delta^{-1}(\Lambda)$, then $x$ is called a {\em conversion}, and $\delta(x)$ is its associate loss. The set of all conversions is denoted by $\Delta'$ and the set of associate losses by $\Lambda'$. The 6-tuple $(G,G',S,\phi,\rho, \delta)$ is called a {\em DLC reconciliation}. 
\end{definition}
We see that every DL reconciliation is also a DLC reconciliation with $\Delta' = \emptyset$.
From now on, \emph{reconciliation} stands for \emph{DLC reconciliation}. Examples of DLC reconciliations are drawn on Figure \ref{fig:reconciliation_examples}.

The following properties are equivalents of standard properties of DL reconciliations \cite{Chauve2008}, which have to be checked in the DLC case.

\begin{lemma}
	Let $\R=(G,G',S,\phi,\rho,\delta)$ be a reconciliation, $x,y\in V(G')$ and $x < y$. Then $\rho(x) \le \rho(y)$.
\end{lemma}
\begin{proof}
	If $x<y$, then we have $x_1,...,x_k \in V(G')$ so that $x=x_0<x_1<x_2<...<x_k<x_{k+1}=y$, and $x_i$ is a child of $x_{i+1}$. From  Definition \ref{def:consis}, we have that $(D)$ or $(S)$ holds, {\em i.e.} $\rho(x) \le \rho(p(x))$, therefore $\rho(x)\le \rho(x_1)\le \rho(x_2) \le ...\le \rho(x_k) \le \rho(y)$.  \qed
\end{proof}
 
\begin{lemma}\label{lem:lowest_is_spec}
	Let $\R=(G,G',S,\phi,\rho, \delta)$ be a reconciliation, $s\in V(S)\backslash L(S)$, $x\in V(G')\backslash L(G')$ such that $\rho(x)=s$. Then $x\in \Sigma(\R)$ if and only if $x$ is a minimal element of $\rho^{-1}(s)$.
\end{lemma}
\begin{proof}
	Let $x$ be a minimal element of $\rho^ {-1}(s)$. Assume the opposite, then $x\in \Delta(\R)$. Let $x_l,x_r$ be the children of $x$ in $G'$, hence $x_l<x$, $x_r<x$ and $\rho(x)= \rho(x_l) = \rho(x_l) = s$, which contradicts the minimality of $x$. 
	
	Let $x\in \Sigma(\R)$. Assume the opposite, that $x$ is not a minimal element of $\rho^{-1}(s)$. Let $x'<x$, $\rho(x')=s$. Then $x'\le x_l$ or $x'\le x_r$. Let $x'\le x_l$, hence $\rho(x')\le \rho(x_l)\le \rho(x)$. Therefore $\rho(x)=\rho(x_l)$, which  contradicts  $x\in \Sigma(\R)$. \qed
\end{proof}

Next lemma states that we cannot have two comparable speciations assigned to the same node from $V(S)$. 
\begin{lemma}\label{lem:no_two_speciat}
	Let $\R=(G,G',S,\phi,\rho, \delta)$ be a reconciliation and $x,y\in V(G')$, $x<y$, $\rho(x)=\rho(y)$. Then  $y\in \Delta(\R)$.
\end{lemma}
\begin{proof}
	Follows directly from Lemma \ref{lem:lowest_is_spec}.  \qed
\end{proof}

\begin{lemma}\label{lem:comparable}
	Let $\R_1=(G,G'_1,S,\phi,\rho_1, \delta_1)$ and $\R_2=(G,G'_2,S,\phi,\rho_2, \delta_2)$ be reconciliations, and $x\in V(G)$. Then $\rho_1(x)$ and $\rho_2(x)$ are comparable.
\end{lemma}
\begin{proof}
	Assume the opposite, {\em i.e.}  $\rho_1(x)$ and $\rho_2(x)$ are incomparable. Then $T(\rho_1(x))$ and $T(\rho_2(x))$ are disjoint, and in particular $L(\rho_1(x))\cap L(\rho_2(x))=\emptyset$. Let $l\in L(x)$. Then $l\le x$, therefore $\phi(l)=\rho_1(l) \le \rho_1(x)$ and $\phi(l)=\rho_2(l) \le \rho_2(x)$, hence $\phi(l)\in L(\rho_1(x))$ and $\phi(l)\in L(\rho_2(x))$, a contradiction.  \qed  
\end{proof}

\begin{definition}[The cost/weight of a reconciliation]
	Let $\R=(G,G',S,\phi,\rho, \delta)$ be a reconciliation, $d,l,c\in \mathbb{N}$ weights associated with duplication, loss and conversion. The cost (or weight) of $\R$ is given by $$\omega(\R)=l\cdot|\Lambda\backslash\Lambda'|+d\cdot|\Delta \backslash \Delta'|+ c\cdot|\Delta'|.$$   
\end{definition}

Examples of computations of this cost are given on Figure \ref{fig:reconciliation_examples}.
As we can see, losses from $\Lambda'$ are not counted as losses in the formula, so we call them \emph{free losses}. If a lost subtree has only free losses then it is called a {\em free subtree}.

\begin{definition}[Minimum/optimal reconciliation]
	Let $\R=(G,G',S,\phi,\rho, \delta)$ be a reconciliation that minimizes $\omega(\R)$, for given $G$, $S$, and $\phi$. Then it is called {\em minimum} (or {\em optimal}) reconciliation.
\end{definition}

In the sequel we give an algorithm that is able to output all optimal reconciliations for $d=l=c$, so unless specified, we assume from now, and without loss of generality, that they are all equal to 1. We come back to the general case in the conclusion, stating open problems.

\subsection{Completions and minimizations of reconciliations}

Recall that any DL reconciliation is a DLC reconciliation by definition. However an optimal DL reconciliation is not an optimal DLC reconciliation. Completions and minimizations are operations on reconciliations that help constructing nonetheless a relation between optimal DL and DLC reconciliations.

\begin{definition}[Loss extension]
Let $\R=(G,G',S,\phi,\rho, \delta)$ be a reconciliation. The reconciliation $\R'=(G,G'',S,\phi,\rho',\delta')$ is said to be obtained from $\R$ by {\em loss extension} if $G''$ is an extension of $G'$, $\rho=\rho'/V(G')$, $\R$ and $\R'$ have the same number of lost subtrees.
\end{definition}

\begin{definition}[Completion]
Let $\R$ be a reconciliation, and $\R'$ is a reconciliation with minimum weight among all reconciliations obtained from $\R$ by extending some losses. Then $\R'$ is called a \emph{completion} of $\R$.   
\end{definition}

It is obvious, by definition, that an optimal reconciliation is a completion, {\em i.e} a completion of a reconciliation $\R$ has always a lower or equal cost than $\R$ itself.
The set of all completions of $\R$ is denoted by $c(\R)$. When useful, $c(\R)$ can also be used to denote one arbitrary completion if it is clear that any completion works. For example the cost of a completion can be written $\omega(c(\R))$ since by definition they all have the same cost.

The converse of a completion is a {\em minimization}. It is based on the following definition and lemma.

\begin{definition}[Minimal reconciliation]  
	A reconciliation $\R=(G,G',S,\phi,\rho, \delta)$ is called {\em minimal} if there does not exist $G''$ such that $G'$ is a proper extension of $G''$, $G''$ is an extension of $G$, and $G''$ is $\rho''$-consistent, where $\rho''=\rho/V(G'')$.      
\end{definition}

An example of minimal reconciliation is the LCA reconciliation. The next lemma shows how to construct a minimal reconciliation from any reconciliation.

\begin{lemma}\label{lem:positions}
	Let $G$ and $S$ be a gene and a species tree,  and  $\rho':V(G)\rightarrow V(S)$ such that 
	\begin{itemize}
		\item $\rho'(x)=\phi(x)$, $\forall x\in L(G)$,
		\item $x<y \implies \rho'(y)\le \rho'(y)$,
		\item $\rho'(x)$ belongs to the path from $\rho_{lca}(x)$ to $root(S)$.
	\end{itemize}
	Then there exists a unique (up to $\delta$) minimal reconciliation $\R=(G,G',S,\phi,\rho, \delta)$ such that $\rho/V(G)=\rho'$.  
\end{lemma}
\begin{proof}
	Assume that there exists a reconciliation $\R_1=(G,G'_1,S,\phi,\rho_1, \delta_1)$ such that $\rho_1/V(G)=\rho'$. Let $x\in V(G)$ with children $x_l,x_r$ (in $G$). In the next three cases we show how to construct $G'$.

Case 1, $\rho_1(x_l)=\rho_1(x)$ and $\rho_1(x_r) < \rho_1(x)$. In that case $x\notin \Sigma(\R_1)$, hence $x\in \Delta(\R_1)$. Therefore $\exists x'\in V(G'_1)$ such that $x'$ is the right child of $x$ and $\rho_1(x')=\rho_1(x)$. Since $x_r<x'<x$, $x'$ is not a leaf and it has the left subtree. Therefore $\exists x''\in V(G'_1)$ such that $x''$ is a descendant of $x'$ and $\rho_1(x'')=\rho_1(x')_l$. We have a similar situation for the case $\rho_1(x_r)=\rho_1(x)$ and $\rho_1(x_l)<\rho_1(x)$.

Case 2,  $e=(s,p(s))\in E(S)$, $s\in V(S)$ and $\rho_1(p_G(x))>s$ and $\rho_1(x)<s$. We will prove that there exists a node $x_1\in V(G'_1)$ such that $\rho(x_1)=s$ and $x<x_1<p_G(x)$. Let $x'$ be a minimal node of $V(G'_1)$ such that $x'\le p_G(x)$ and $\rho(x')>s$. From Lemma \ref{lem:lowest_is_spec}, we have $x'\in \Sigma(\R_1)$. Therefore it has children $x'_l$,$x'_r$ (in $G'_1$) such that $\rho_1(x'_l)<\rho_1(x')$ and $\rho_1(x'_r)<\rho_1(x')$. From the properties of $x'$, we get that one of the children maps to $s$. Let $\rho(x'_r)=s$, and we need to insert an additional child for $x'_r$, since $x'_r$ cannot be a leaf.  

Case 3, $\rho_1(x_l)\le \rho_1(x)_l$ and $\rho_1(x_r)\le \rho_1(x)_l$. Let $x'$ be a child of $x$ in $G'_1$. Therefore $x'$ is comparable to $x_l$ or $x_r$, and $\rho_1(x')$ is comparable to $\rho_1(x_l)$ or $\rho_1(x_r)$, hence $\rho_1(x')$ is comparable to $\rho_1(x)_l$. Next, $\rho_1(x')$ is incomparable to $\rho_1(x)_r$, hence $x\notin \Sigma(\R_1)$ and $x\in \Delta(\R_1)$. If $x'_l,x'_r$ are the children of $x$ in $G'_1$, then $\rho_1(x'_l)=\rho_1(x'_r)=\rho_1(x)$. This means that we need to insert $x'_l,x'_r$ and additional children for $x'_l,x'_r$.   

Insertions, described in the previous three cases, are for any reconciliation $\R_1$. Let us prove that they are enough to form a reconciliation. From this will follow minimization and uniqueness. 

Let us form $G'$ and $\rho$ in a way described in the previous three cases. We need to prove that $G'$ is $\rho$-consistent. Let $x\in V(G')\backslash L(G')$ and $x_l,x_r$ are the children of $x$ in $G'$. We will prove that $x$ satisfies condition $(D)$ or $(S)$ from Definition \ref{def:consis}. If $\rho(x)=\rho(x_l)=\rho(x_r)$, then condition $(D)$ is satisfied. Now assume that condition $(D)$ is not satisfied, i.e. $\rho(x)\ne \rho(x_l)$ or $\rho(x)\ne \rho(x_r)$. Take $\rho(x_r) < \rho(x)$. From the Case 2, we get $\rho(x_r) = \rho(x)_r$. We are left to prove $\rho(x_l) = \rho(x)_l$. Assume the opposite, let $\rho(x_l)= \rho(x)_r$ or $\rho(x_l)=\rho(x)$. From Case 3 and the definition of duplication, we get that $x$ is a duplication, this contradicts our assumption that $\rho(x_l) \ne \rho(x)_l$. \qed
\end{proof}

The unique minimal reconciliation obtained from a reconciliation is called its {\em minimization}. In the next section we prove that minimization and completion are complementary operations, that is, an optimal reconciliation is always the completion of its minimization. This will lead to the important result that completions of the LCA reconciliations are optimal.


\section{A family of optimal reconciliations: LCA reconciliations}\label{sec:LCA}

In this section we provide a polynomial running time algorithm which finds an LCA completion, and prove that it is an optimal reconciliation. We present a more general algorithm, which finds a completion of any reconciliation. To this aim we present the important notion of {\em flow}, constantly used all along the paper. This settles the complexity of the defined problem when the weights $d,l,c$ are all equal. However the algorithm described here does not find all LCA completions, and moreover the space of optimal reconciliations is not limited to LCA completions. Finding all solutions will be the subject of next section.
Here we begin by stating general properties of reconciliations and optimal reconciliations, showing that they all share some important properties with LCA reconciliations.


\subsection{Similarities of any reconciliation with the LCA reconciliation}

Some properties of the LCA reconciliation are shared by all reconciliations.

\begin{lemma}\label{lem:notunder}
	Let $\R=(G,G',S,\phi,\rho, \delta)$ be a reconciliation, and $x\in V(G)$. Then $\rho(x)$ is not lower than $\rho_{lca}(x)$.  
\end{lemma}
\begin{proof}
	Follows directly from the definition of Last Common Ancestor. \qed
\end{proof}

\begin{lemma}\label{lem:lca->root(S)}
	Let $\R=(G,G',S,\phi,\rho, \delta)$ be a reconciliation, and $x\in V(G)\backslash L(G)$. Then $\rho(x)$ is in the path in $S$ from $\rho_{lca}(x)$ to $root(S)$. 
\end{lemma}
\begin{proof}
	Follows directly from Lemmas \ref{lem:comparable} and \ref{lem:notunder}. \qed
\end{proof}

The next lemma states that if a node is a speciation in an arbitrary reconciliation then it is also a speciation in the LCA.

\begin{lemma}\label{lem:spec}
	Let $\R=(G,G',S,\phi,\rho, \delta)$ be a reconciliation, and $x \in V(G)$. If $x\in \Sigma(\R)$, then $x\in \Sigma(\R_{lca})$, and $\rho(x)=\rho_{lca}(x)$.        
\end{lemma}
\begin{proof}	
	Let $x\in V(G)\cap \Sigma(\R)$. Let $x''_l,x''_r$ be the children of $x$ in $\R$, $x'_l,x'_r$ the children of $x$ in $\R_{lca}$, and $x_l,x_r$ be the children of $x$ in $G$. We have $\rho(x)_l=\rho(x''_l)$ and $\rho(x)_r=\rho(x''_r)$. From Lemma \ref{lem:lca->root(S)} we have $\rho_{lca}(x)\le \rho(x)$.
	
	Assume that $\rho_{lca}(x)<\rho(x)$. Hence $\rho(x)_l$ or $\rho(x)_r$ is incomparable to $\rho_{lca}(x)$. Assume that $\rho(x)_r=\rho(x''_r)$ is incomparable to $\rho_{lca}(x)$. Next, $x_r\le x'_r<x$, $x_r\le x''_r<x$, hence $\rho_{lca}(x_r)\le \rho_{lca}(x'_r)\le \rho_{lca}(x)$ and $\rho(x_r)\le \rho(x''_r)\le \rho(x)$. Therefore, $\rho(x_r)$ is incomparable to $\rho_{lca}(x)$, hence incomparable to $\rho_{lca}(x_r)$, which contradicts Lemma \ref{lem:comparable}. Therefore $\rho_{lca}(x)=\rho(x)$. 
	
	Let us prove that $x\in \Sigma(\R_{lca})$. Assume the opposite, $x\in \Delta(\R_{lca})$. Thus $\rho_{lca}(x)=\rho_{lca}(x'_l)=\rho_{lca}(x'_r)$, and from LCA reconciliation, we have $\rho_{lca}(x)=\rho_{lca}(x_r)$ or $\rho_{lca}(x)=\rho_{lca}(x_l)$. Next, $\rho_{lca}(x_r)=\rho_{lca}(x)=\rho(x)>\rho(x_r)$ or $\rho_{lca}(x_l)=\rho_{lca}(x)=\rho(x)>\rho(x_l)$, which contradicts Lemma $\ref{lem:notunder}$. \qed   
\end{proof}

Thanks to these properties we can define a distance from an arbitrary reconciliation to the LCA reconciliation. This distance will be used in the proofs of several properties, stating that there is always a way to lower the distance to the LCA without increasing the cost of a reconciliation.
\begin{definition}
Let $\R=(G,G',S,\phi,\rho, \delta)$ be any reconciliation.
Let $dist_{lca}(\R)=\sum_{d\in V(G)} d_S(\rho(d),\rho_{lca}(d))$ be the {\em distance} from $\R$ to the LCA reconciliation $\R_{lca} = (G,G'_{lca},S,\phi,\rho_{lca})$.
\end{definition}

\begin{lemma} \label{lem:duplic_lower}
  If for a reconciliation $\R$ $dist_{lca}(\R)>0$, there exists a reconciliation $\R'$ such that $dist_{lca}(\R')<dist_{lca}(\R)$ and $\omega(\R')\le \omega(\R)$. 
\end{lemma}
\begin{proof}
	Take any $d'\in V(G)$ so that $\rho(d') > \rho_{lca}(d')$ and let $d$ be a minimal element of $V(G)$ such that $\rho(d)=\rho(d')$ and $d \le d'$. Since $d\le d'$, we have $\rho_{lca}(d) \le \rho_{lca}(d')<\rho(d')=\rho(d)$, therefore $\rho_{lca}(d)<\rho(d)$.
	By Lemma \ref{lem:spec} $d\notin \Sigma(\R)$, so $d\in \Delta(\R)$. 
	
	Let $d^1_l,d^1_r$ be the children of $d$ in $\R$. Since $d\in \Delta(\R)$, we have $\rho(d)=\rho(d^1_l)=\rho(d^1_r)$, and because of the minimality of $d$, we get $d^1_l,d^1_r \notin V(G)$. Similarly, all descendants of $d$ in $G'$, with the same $\rho$-value, are not in $V(G)$. 
	
	Let $d_1,...d_k$ be these descendants and let $T_1,...,T_k$ be lost subtrees such that $root(T_i)=d_i$, ($i=1,\ldots, k$). Prune all these subtrees, contract nodes of degree two ({\em i.e.} $d_1,\ldots, d_k$), and let $G''$ denotes the obtained extension of gene tree $G$. Let $d^2_l,d^2_r$ be the children of $d$ in $G''$. 
	
	If $\rho(d^2_l) \ne \rho(d^2_r)$, then $G''$ generates a new reconciliation $\R'$, where $d$ is a speciation, and $\rho'(d)=\rho(d)$. By Lemma $\ref{lem:spec}$, $\rho'(d)=\rho_{lca}(d)$, which contradicts $\rho(d)>\rho_{lca}(d)$.  
	
	Let $\rho(d^2_l) = \rho(d^2_r)$. Since $\rho(d^2_l) <\rho(d)$, we don't have consistency. Put $\rho'(d)=\rho(d^2_l)$ and insert $x_1$ into $G''$ so that $d<x_1<p_{G'}(d)$, $\rho'(x_1)=\rho(d)$, and $x_1$ is the root of some of the pruned subtrees $T_i$ (reinsert $T_i$). In this way we get a new reconciliation $\R''$, and $d$ is a duplication in $\R''$. Also $\omega(\R'')\le \omega(\R)$ and $dist_{lca}(\R'') < dist_{lca}(\R)$.
	
	If $d\in  \Delta'(\R)$ and corresponding loss is $l$, then extend $l$ so that one loss extensions follows $d$ and the other can be some of the pruned subtrees $T_i$ (reinsert $T_i$). \qed
\end{proof}

The next lemma states that with LCA we get the smallest set of duplications.
\begin{lemma}\label{lem:LCA_duplic}
	Let $\R_{lca}$ be the LCA reconciliation and $\R$ be any reconciliation. Then $\Delta(\R_{lca}) \subseteq \Delta(\R)\cap V(G)$.         
\end{lemma}
\begin{proof}
	Let $x\in \Delta(\R_{lca})$, then $x \notin \Sigma(\R_{lca})$ and  $x\in V(G)$. Assume the opposite, that $x\notin \Delta(\R)\cap V(G)$, then $x\in \Sigma(\R)$. From Lemma $\ref{lem:spec}$ we get $x\in \Sigma(\R_{lca})$, a contradiction. Therefore $x\in \Delta(\R) \cap V(G)$.  \qed 
\end{proof}

\subsection{Properties of optimal reconciliations}

We examine some properties of optimal reconciliations. Note that optimal reconciliations are not necessarily minimal, but we will state the relation between the two classes (see Lemma \ref{lem:optimal->minimal}).
The next lemma states that optimal reconciliations never contain duplication nodes in lost subtrees.

\begin{lemma} \label{lem:duplic_in_G}
	Let $\R=(G,G',S,\phi,\rho, \delta)$ be an optimal reconciliation. Then $\Delta(\R) \subseteq V(G)$, {\em i.e.} all duplications nodes are in $G$. 	
\end{lemma}
\begin{proof}
	Assume the opposite. Let $\R=(G,G',S,\phi,\rho, \delta)$ be a reconciliation, and $x$ is a minimal node of $\Delta(\R) \backslash V(G)$. Let us prove that $\R$ cannot be optimal. Let $x_l,x_r\in V(G')$ be the children of $x$. Since $x$ is a duplication, we have $\rho(x)=\rho(x_l)=\rho(x_r)$. Observe two cases.
    
	Case 1, $x_l,x_r\notin V(G)$
    
	Case 1.1, $x$ is a conversion, and  $l$ is the corresponding loss. Remove $l$ and $x$, connect $x_l$ with $p_{G'}(l)$, and $x_r$ with $p_{G'}(x)$. In this way we get $G''$. Let $\rho'=\rho/G''$, and $\delta'=\delta/G''$. We get a reconciliation $\R'=(G,G'',S,\phi,\rho', \delta')$ which has one duplication less, {\em i.e.} $\omega(\R')=\omega(\R)-1$. Hence $\R$ cannot be an optimal reconciliation.
	
	Case 1.2, $x$ is not a conversion. Remove $T(x_l)$ and $x$, then connect $x_r$ with $p_{G'}(x)$. By a similar argument, we get a reconciliation with one duplication and all non-free losses from $T(x_l)$ less, {\em i.e.} we get a reconciliation with a strictly lower cost. Indeed, since $x$ is a minimal duplication, subtree $T(x_l)$ cannot have any duplications, {\em i.e.} by removing $T(x_l)$ we cannot get to the situation where some free loss becomes non-free.

Case 2, $x_l\in V(G), x_r\notin V(G)$. Similarly, if $x$ is not a conversion, remove $T(x_r)$ suppress $x$, and we get a reconciliation with strictly less cost. If $x$ is a conversion and $l$ is associate loss, then remove $l$, suppress $x$ and connect $x_r$ and $p_{G'}(l)$. We again obtain a cheaper reconciliation. \qed
\end{proof}

The next lemma is a version of  Lemma \ref{lem:duplic_lower} for an optimal reconciliation.  
\begin{lemma}\label{lem:lowering_duplic}
	Let $\R_{lca}$ be the LCA reconciliation, and let $\R$ be an optimal reconciliation. If $dist_{lca}(\R)>0$, there exists an optimal reconciliation $\R'$ such that $\Delta(\R')=\Delta(\R)$ and $dist_{lca}(\R')<dist_{lca}(\R)$. 
\end{lemma}
\begin{proof}
	Follows directly from the proof of Lemma $\ref{lem:duplic_lower}$. We constructed $\R'$ by pruning some of the lost subtrees and lowering duplication, which remained a duplication in $\R'$. By Lemma $\ref{lem:duplic_in_G}$ lost subtrees in optimal reconciliation cannot contain duplications, hence the set of duplications remained unchanged, {\em i.e.} $\Delta(\R')=\Delta(\R)$. \qed  
\end{proof}

Next theorem states that all optimal reconciliations have the same sets of duplications.  
\begin{theorem}\label{th:duplic_const}
	Let $\R_{lca}=(G,G_{lca}',S,\phi,\rho_{lca})$ be the LCA reconciliation and $\R=(G,G',S,\phi,\rho, \delta)$ be an optimal reconciliation. Then $\Delta(\R_{lca}) = \Delta(\R)$.      
\end{theorem}
\begin{proof}
	Assume the opposite, there exist $G$, $S$ and $\R$ such that $\R$ is an optimal reconciliation and $\Delta(\R_{lca}) \ne \Delta(\R)$. By Lemma $\ref{lem:LCA_duplic}$ and Lemma $\ref{lem:duplic_in_G}$ we get $\Delta(\R_{lca}) \subset \Delta(\R)\cap V(G)=\Delta(\R)$. Assume that $\R$ is an optimal reconciliation with $\Delta(\R_{lca}) \subset \Delta(\R)$ and minimum $dist_{lca}(\R)$. We have $dist_{lca}(\R)=0$, otherwise we could get an optimal reconciliation $\R'$ with $dist_{lca}(\R') <dist_{lca}(\R)$ and $\Delta(\R')=\Delta(\R)$ (Lemma $\ref{lem:lowering_duplic}$). From $dist_{lca}(\R)=0$, we obtain $\rho(x)=\rho_{lca}(x)$, $\forall x\in V(G)$. 
	
	Let $x'\in \Delta(\R)\backslash \Delta(\R_{lca})$. By Lemma $\ref{lem:duplic_in_G}$, we have $x' \in V(G)$. From $x'\notin \Delta(\R_{lca})$ we get that $x'\in \Sigma(\R_{lca})$. We will continue in a similar way as in the proof of Lemma $\ref{lem:duplic_lower}$. Let $x_1,...,x_k$ be descendants of $x'$ in $V(G')$ with the same $\rho$-value as $x'$. 
	
	Assume $x_1\in V(G)$. Since $\rho(x)=\rho_{lca}(x)$, $\forall x\in V(G)$ and $\rho(x_1)=\rho(x')$ we get $\rho_{lca}(x_1)=\rho_{lca}(x')$, hence (Lemma \ref{lem:no_two_speciat}) $x'\in \Delta(\R_{lca})$, a contradiction. Therefore $x_1\notin V(G)$.

	By a similar argument, $x_1,\ldots,x_k\notin V(G)$. Let $T_i$ be the lost subtrees rooted at $x_i$ $(i=1,\ldots,k)$. By pruning $T_i$ and suppressing $x_i$ $(i=1,\ldots,k)$ we get $G''$, and a new reconciliation where node $x'$ is a speciation. Hence we get a reconciliation with strictly lower cost, which contradicts the optimality of $\R$.        \qed 		
\end{proof}

Next lemma states that, in an optimal reconciliation, we cannot have two comparable nodes $x,y\in V(G')\backslash V(G)$ such that $\rho(x)=\rho(y)$.
\begin{lemma} \label{lem:min_added_losses}
	Let $\R$ be an optimal reconciliation and $x,y\in V(G')$ such that $\rho(x)=\rho(y)$ and $x<y$. Then $y\in V(G)\cap \Delta(R_ {lca})=\Delta(\R_{lca})=\Delta(\R)$.   
\end{lemma}
\begin{proof}
	From Lemma \ref{lem:no_two_speciat} we have $y\in \Delta(\R)$. From Theorem \ref{th:duplic_const}, we obtain $\Delta(\R)=\Delta(\R_{lca})$. From Lemma \ref{lem:duplic_in_G}, we have $y\in V(G)\supseteq \Delta(\R)=\Delta(\R_{lca})$. Therefore $y\in V(G)\cap \Delta(\R_{lca})$. \qed 	
\end{proof}

Next lemma states the relation between minimal and optimal reconciliations.

\begin{lemma}\label{lem:optimal->minimal}
	Let $\R$ be an optimal reconciliation. Then there exists $\R'$, a minimal reconciliation such that $\R$ is a completion of $\R'$.
\end{lemma}
\begin{proof}
	Let $\R'$ be the reconciliation obtained from $\R$ by deleting all lost subtrees except their root edges. So $\R$ is a completion of $\R'$. We prove that $\R'$ is minimal. Suppose the opposite.
    There is $e'=(x',p_{G'}(x'))\in E(G')\backslash E(G)$ such that by removing $e'$ and suppressing $p_{G'}(x')$ we obtain again a reconciliation, denoted by $\R''$. From the proof of Lemma \ref{lem:positions}, Case 2,
    we have that $\forall s\in V(S)$ and $x,y\in V(G')$, such that $x<y$, $\rho(x)<s<\rho(y)$, $\exists z\in V(G')$ such that $\rho(z)=s$ and $x<z<y$. Let $x_1$ be another child of $p_{G'}(x')$. Since there is no lost subtrees with more than one edge, we have $x_1\in V(G)$. 
	
	Let $s=\rho(p_{G'}(x'))$. Since $\R''$ is a reconciliation, $\exists x''\in V(G'')$ such that  $s=\rho(x'')$ and $x''$ comparable to $x_1$. Take minimal $x''$ with these properties, then (Lemma \ref{lem:lowest_is_spec}) $x''\in\Sigma(\R'')$. After bringing back $e$, we get that $p_{G'}(x')$ or $x''$ becomes a duplication (Lemma \ref{lem:no_two_speciat}). Hence $\Delta(\R'')\subset \Delta(\R')=\Delta(\R)$, which contradicts the optimality of $\R$ (Lemma \ref{lem:LCA_duplic} and Theorem \ref{th:duplic_const}).  \qed
\end{proof}

\subsection{LCA completions are optimal}

\begin{theorem}\label{th:LCAopt}
A completion of the LCA reconciliation is an optimal reconciliation.      
\end{theorem}
\begin{proof}
Let $\R=(G,G',S,\phi,\rho, \delta)$ be an optimal reconciliation with 
$dist_{lca}(\R)$ minimum. We prove that this reconciliation is a completion of the LCA. Since all completions of the LCA have the same weight by definition, this proves that all completions of the LCA are optimal reconciliations.

From Lemma  $\ref{lem:lowering_duplic}$ we get $dist_{lca}(\R)=0$ and therefore $\rho(x)=\rho_{lca}(x)$, $\forall x\in V(G)$. From Theorem $\ref{th:duplic_const}$ and Lemma $\ref{lem:duplic_in_G}$, we have $\Delta(\R)=\Delta(\R_{lca})\subseteq V(G)$.    
	
	Let $t$ be a root of some lost subtree of $G'$. Let us prove that $t\in V(G'_{lca})$, and vice versa, if $t\in  V(G'_{lca})\backslash V(G)$, then $t$ is a root of some lost subtree of $G'$. This correspondence has to be bijective.  
	
	Let us prove that we can establish a bijection \\
    $f:V(G)\cup \{t \mid \text{t is a root of some lost subtree of } V(G')\} \rightarrow V(G'_{lca})\backslash \Lambda(\R_{lca})$ such that $f(x)=x$, $\forall x\in V(G)$, $x<y\implies f(x)<f(y)$, $\rho(x)=\rho(f(x))$.    
	
	First, put $f(x)=x$, $\forall x\in V(G)$.
	
	Let $t\in V(G')\backslash V(G)$ be a root of some lost subtree of $G'$, $\rho(t)=s$, $x<t<p_G(x)$. From Lemmas \ref{lem:duplic_in_G} and \ref{lem:lowest_is_spec}, we have $t\in\Sigma(\R)$ and $t$ is a minimal element of $\rho^{-1}(s)$. Hence, there is no other element $t'\in V(G')$ such that  $\rho(t')=s$, $x<t'<p_G(x)$. Since $t\in \Sigma(\R)$, we have $\rho(x)<\rho(t)\le \rho(p_G(x))$. In $\R_{lca}$ we also have $x'\in V(G'_{lca})$, such that $\rho(x')=s$, and $x<x'<p_G(x)$. Next, put $f(t)=x'$. 
	
	Above correspondence is obviously an injection. Let us prove that it is a surjection. In a similar way, let $x'\in V(G'_{lca})\backslash \Lambda(\R_{lca})$, $\rho_{lca}(x')=s'$. If $x'\in V(G)$, then $x'=f(x')$. Now, assume $x'\notin V(G)$.  Again from Lemmas \ref{lem:duplic_in_G} and \ref{lem:lowest_is_spec} we have that $x'\in \Sigma(\R_{lca})$ and $x'$ is a minimal element of $\rho_{lca}^{-1}(s')$. Let $x<x'<p_G(x)$, $x\in V(G)$. Similarly, we have $\rho_{lca}(x)<\rho_{lca}(p_G(x))$ and $x'$ is the only element from $V(G')\backslash V(G)$ assigned to $s'$ comparable to $x$.  In order for $\R$ to be $\rho$-consistent, there is  a root of the lost subtree of $G'$ (say $t$) such that: $\rho(t)=s'$, and $x<t<p_G(x)$ and it is unique. So, $f(t)=x'$.
	
	We proved the existence of the described correspondence, therefore every lost subtree of $\R$ is obtained as a loss extension in $\R_{lca}$. \qed           
\end{proof} 

The LCA reconciliation is easy to find, it is a well known result that there is a linear time algorithm to compute it \cite{Chauve2009}. What remains in order to derive an algorithm to find an optimal reconciliation is to find a completion. Next section presents a method to find a completion of an arbitrary reconciliation.


\subsection{Finding a completion and the flow of losses}\label{sec:flow}

Finding a completion is a kind of flow problem. We have demands, which are losses, that we supply by duplications, {\em i.e.} we associate them to duplications to form conversions. The amount and distribution of duplications in the phylogenetic tree tells how many losses can be supplied. The number of losses that can be supplied tells the value of a completion. We compute this number recursively along the tree. In consequence we have to define restriction of reconciliations to subtrees, which are {\em multiple reconciliations}.

\begin{definition}[Multiple reconciliation]\label{def:multRecon}
Let $\R_i=(G_i,G'_i,S,\phi_i,\rho_i)$ be DL reconciliations of gene trees $G_i$ with species tree $S$, ($i=1,\ldots,k$).
Let $T_1,\ldots,T_t$ be trees, $\rho'_j:V(T_j)\rightarrow V(S)$ verifying that $\rho'_j(root(T_j))=root(S)$ and $T_j$ is $\rho'_j$-consistent, ($j=1,\ldots,t$). Let $\R'_j=(T_j,S,\rho'_j)$,  ($j=1,\ldots,t$).
Next, let $\delta: \bigcup\Delta(\R_i) \cup \bigcup\Delta(\R'_j) \rightarrow \bigcup\Lambda(\R_i) \cup \bigcup\Lambda(\R'_j)$ be a partial injective function such that $\delta(d)=l$ implies that $d$ and $l$ are assigned to the same node in $V(S)$. Then the structure $\R_{m}=(G,S,\R_1,...,\R_k,\R'_1,...,\R'_t,\delta)$ is called {\em multiple reconciliation}.
\end{definition} 

The crucial property of a multiple reconciliation is that a loss from one tree ($G'$ or $T_i$) can be assigned by $\delta$ to a duplication from another gene tree. The cost of a multiple reconciliation is computed the same way as the cost of a reconciliation. The multiple reconciliation {\em induced} by a reconciliation $\R$ and an edge $e$ is composed of all parts of $\R$ mapped to $S(e)$ by $\rho$.
If it is evident from the context, instead of \emph{multiple reconciliation}, we will write \emph{reconciliation}, allowing additional lost subtrees.
Let $\R_m$ be a multiple reconciliation with $e\in E(S)$. Let $\R_{m1}$ be the reconciliation obtained from $\R_m$ by adding $k$ new lost subtrees with only one root edge assigned to $e$. Obviously $\omega(\R_m)+k=\omega(\R_{m1})$, but it is possible that $\omega(c(\R_m))=\omega(c(\R_{m1}))$ (see Figure \ref{fig:flow}).

\begin{definition}[Flow]
Let $\R$ be a reconciliation, $e\in E(S)$, and $\R(e)$ the multiple reconciliation induced with $\R$ and $e$. Let $\R'(e)$ be the reconciliation obtained from $\R(e)$ by removing all $T_1,\dots,T_l$ the lost trees containing only one loss assigned to $e$. With $\R_k(e)$ denote multiple reconciliation obtained from $\R'(e)$ by adding $k$ lost trees containing only one loss assigned to $e$ ($k$ may be lower or higher than $l$, if $k=l$ then $\R_k=\R$). Let $k'$ be the maximum number such that $\omega(c(\R_{k'}(e)))=\omega(c(\R'(e)))$. With $F(e,\R)=F(e)=k'-l$ is denoted the {\em flow}  of the edge $e$.  
\end{definition}
Note that if $F(e) \ge 0$, then $F(e)$ is the maximum number of extra  losses assigned to $e$ that does not change the weight of the completion of $\R(e)$. Opposite is also true, if $m\ge 0$ is the maximum number of extra losses assigned to $e$ that does not change the weight of a completion of $\R(e)$, then $m=F(e)$.

\begin{figure}[h]
    \centering
    \includegraphics{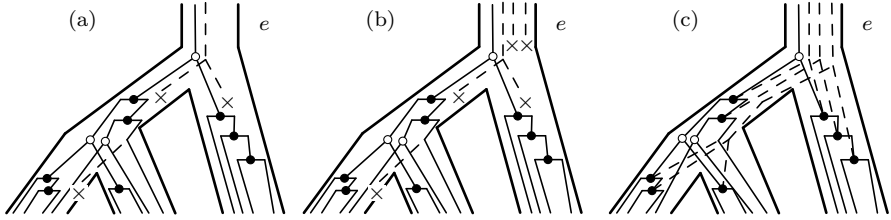}
	\caption{Flow. (a) Multiple reconciliation $\R_1$. (b) Multiple reconciliation $\R_2$ obtained from $\R_1$ by adding $k$ extra ($k=2$) losses to the edge $e$. We have $\omega(\R_2)=\omega(\R_1)+k$. (c) Completion of $\R_2$. Completion of $\R_1$ can be obtained by removing added lost subtrees. We see that $\omega(c(\R_2))=\omega(c(\R_1))$.
    Maximum number $k$ for which the last relation holds is called flow of the edge $e$}
    \label{fig:flow}
\end{figure}

We show how to efficiently compute the flow recursively with Lemma \ref{lem:F(e)_calculation}.
Recall $\D(e)=\D(e,\R_m)$, $\L(e)=\L(e,\R_m)$ denote number of duplications and losses assigned to $e$ in reconciliation $\R_m$.

\begin{lemma}\label{lem:F(e)_calculation}
Let $\R_m$ be a multiple reconciliation, $e\in E(S)$. Then  $$F(e)=max\Big(  min\Big( F(e_l),F(e_r) \Big),  0\Big)+\D(e)-\L(e).$$
\end{lemma}
\begin{proof}
We will use mathematical induction on $e$. Let $e$ be a leaf edge. Then $e_l=NULL$, $e_r=NULL$, and $F(e_l)=F(e_r)=0$. The only way new losses, assigned to $e$, can be free is by pairing them with duplications in $e$. Therefore $k'=d$ and $F(e)=k'-l=d-l$.

Now, let $e$ be a non-leaf edge, $m=max(min(F(e_l,\R_m(e_l)),F(e_r,\R_m(e_r))), 0)$, $d=\D(e,\R_m(e))$, and $l=\L(e,\R_m(e))$. By inductive hypothesis, we can extend $m$ losses over $e_r$ and $e_l$, so the weight of the completions of $\R_m(e_l)$ and $\R_m(e_r)$ is not changed. We can make $d$ losses, assigned to $e$, free by pairing them with duplications in $e$. Hence $k'=m+d$ and $F(e)=k'-l=m+d-l$. \qed
\end{proof}

\begin{lemma}\label{lem:FLessThanZero}
  Let $\R_1$ be a multiple reconciliation with a root edge $e=(s,p(s))$, and $F(e,\R_1)\le 0$. By assigning an extra loss to $e$ we obtain $\R_2$. Then $\omega(c(\R_2))=\omega(c(\R_1))+1$.
\end{lemma}

We postpone the proof of this Lemma to section \ref{sec:algorithm} because it will use some notions introduced later.

The next lemma is a consequence of Lemma \ref{lem:FLessThanZero}.
\begin{lemma}\label{lem:lossRemoved}
 Let $\R_1$ be a (multiple) reconciliation, $e$ is the root edge of $S$, and $F(e,\R_1) < 0$. Let $\R_2$ be a reconciliation obtained from $\R_1$ by removing a loss assigned to $e$. Then $\omega(c(\R_2))=\omega(c(\R_1))-1$
\end{lemma}

Lemmas \ref{lem:FLessThanZero} and \ref{lem:lossRemoved} are stated in a way of adding and removing a loss from the root edge $e$. Similar lemmas are in effect if we remove/add a duplication from/to the root edge $e$. Because of the obviousness we will not state them nor prove them.

Thanks to this flow computation we can find a completion of any reconciliation by a polynomial time algorithm, which pseudo-code is written in Algorithm~\ref{ExtendFreeLosses} and \ref{ExtendLoss}.

 \begin{algorithm}
	\caption{Find a completion of a reconciliation.}\label{ExtendFreeLosses}
	\begin{algorithmic}[1]
	  \Procedure{OneCompletion}{$\R$}	
        \While {there is a loss $l\in\Sigma\backslash\Sigma'$ assigned to edge $e$ such that either there is a duplication that is not a conversion assigned to $e$ or children of $e$ have positive flow}
                 \State \textproc{ExtendLossIntoFreeTree}$(\R,l)$
		      \EndWhile		
		\EndProcedure
	\end{algorithmic}
\end{algorithm}

\begin{algorithm}
	\caption{Extends one loss into a free tree.}\label{ExtendLoss}
	\begin{algorithmic}[1]
		\Procedure{ExtendLossIntoFreeTree}{$\R,l$}		    
          \State $l$ is assigned to $e=(s,p(s))$ and $e_1,e_2$ are children of $e$
          \State $\Delta''(e)$ is the set of all duplications that are not conversion assigned to $e$
          \If {$\Delta''(e)\not=\emptyset$ and $F(e_1)>0$ and $F(e_2)>0$}
          	\State Randomly choose between "assign" and "extend"
           \EndIf
          \If {$F(e_1)\le 0$ or $F(e_2)\le 0$ or "assign" has been chosen}
             	\State Assign $l$ to a random $d\in \Delta''(e)$
           \Else
                \State extend $l$ over $e_1,e_2$
                \State $l_1,l_2$ are new losses assigned to $e_1,e_2$ 
                \State \textproc{ExtendLossIntoFreeTree}$(\R,l_1)$
                \State \textproc{ExtendLossIntoFreeTree}$(\R,l_2)$
             \EndIf
		\EndProcedure
	\end{algorithmic}
\end{algorithm}

Let us introduce a convention. If we say that, {\em e.g.} \emph{$\R'$ is an output of \textproc{ExtendLosses}$(\R)$}, then the procedure \textproc{ExtendLosses}(.) is observed as a standalone procedure with the input $\R$. But if we say that \emph{$\R'$ is an output of \textproc{ExtendLosses}} (no input parameters), then we observe \textproc{ExtendLosses} as a part (sub procedure) of the main procedure, and \textproc{ExtendLosses} receives parameters as described. 

\begin{lemma}\label{lem:ExtendLossIntoFreeTree}
	Let $\R$ be a reconciliation, $l$ is a non-free loss assigned to $e\in E(S)$, $e_1,e_2$ are children of $e$. Next, $\Delta''(e)\ne \emptyset$ or $(F(e_1)>0$ and $F(e_2)>0)$.  Then the procedure \textproc{ExtendLossIntoFreeTree}$(\R,l)$  extends $l$ into a free tree.
\end{lemma}
\begin{proof}
	Note that if $\Delta''(e)=\emptyset$ and $F(e_1)=F(e_2)=0$, then $F(e)\le 0$.
	
	We will use mathematical induction on $e$. Let $e$ be a leaf edge. Then $e_1=NULL$, $e_2=NULL$ and $F(e_1)=F(e_2)=0$. Hence $\Delta''(e)\ne \emptyset$, and $l$ is assigned to a random duplication from $\Delta''(e)$.
	
	Assume that $e$ is not a leaf edge. If $\Delta''(e) \ne \emptyset$ and $assign$ is chosen, then $l$ is assigned to a random element from $\Delta''(e)$, {\em i.e.} $l$ is extended into a free tree with one edge. If $\Delta''(e)=\emptyset$ or $extend$ is chosen, then $F(e_1)>0$, $F(e_2)>0$ and $l$ is extended into $l_1$ and $l_2$. Since $F(e_i)>0$, $(i=1,2)$ then $e_i$ satisfies {\em if} condition in \textproc{OneCompletion}. Hence, by inductive hypothesis, \textproc{ExtendLossIntoFreeTree}$(\R,l_i)$ extends $l_i$ into a free tree, {\em i.e.} $l$ is extended into free tree.   \qed
\end{proof}

Let us introduce a convention. Let $e=(x,p_{G'}(x))\in E(G')$. If $\rho(p_{G'}(x))=p(\rho(x))$, then we can write $\rho(e)=(\rho(x),\rho(p_{G'}(x)))\in E(S)$. This property does not hold for any edge of $G'$, but it holds for any edge of a lost subtree, since we do not observe lost subtrees with duplications (an optimal reconciliation cannot have a lost subtree with a duplication). Let $T$ be a subtree of $G'$, then $\rho(t)=\{\rho(e)\mid e\in E(T)\}$. Sometimes we will identify lost trees with their root, {\em i.e.} $v$ can denote both a root of a tree or a tree with root $v$. The reason for this is that lost subtrees are dynamical, they extend or switch (an operation introduced later), but their roots are not.

\begin{lemma}\label{lem:CompletionOutput_1}
	Let $\R$ be a reconciliation with non-extended losses, $t_i$ ($i=1,\ldots, k$) and $t'_j$ ($j=1,\ldots, m$) are free and non-free lost subtrees of $c(\R)$ such that $t'_j\ge t_i$ whenever $t_{i}$ and $t'_{j}$ overlap. All non-free lost subtrees $t'_j$ ($j=1,\ldots, m$) are non-extended, {\em i.e.} they have one edge each. Then $c(\R)$ is a possible output of  \textproc{OneCompletion}$(\R)$.   
\end{lemma}
\begin{proof}
	Let $\R_0=\R$, $\R_i$ is obtained from $\R_{i-1}$ by extending corresponding loss to the tree $t_i$ ($i=1,\ldots, k$). Hence $\R_k=c(\R)$.
	
	Assume that trees $t_1,\ldots, t_{i-1}$ ($i\ge 1$) are constructed by iterations of  \textproc{ExtendLossIntoFreeTree}. Take $t_i$ that has the minimal root among free lost subtrees that are not added. Let us prove that $F(e,\R_{i-1})>0$, $\forall e\in E(\rho(t_i))\backslash\{root_E(\rho(t_i))\}$. Assume the opposite, let $F(e_1,\R_{i-1})\le 0$, and since free subtree $t_i$ extends over $e_1$, we have that some loss in $S(e_1)$ becomes non-free. More precisely, $\omega(c(\R_{i-1}(e_1))) < \omega(c(\R_i(e_1)))$. This means that $|\Lambda \backslash \Lambda'(c(\R_{i-1}(e_1)))| < |\Lambda \backslash \Lambda'(c(\R_{i}(e_1)))|$. Since trees $t_1,\ldots, t_{i-1}$ (and $t_i$) are free and already present in $\R_{i-1}$ (i.e. $\R_i$), then we can assume that they are not changed in $c(\R_{i-1})$ (i.e. $c(\R_i)$), because we gain nothing by further extending free losses (although it is possible).  
	
	Observe $c(\R_i(e_1))$. Let $T_S$ be the maximal subtree of $S(e_1)$ (see Figure \ref{fig:maximal_subtree}) such that if $v_0 \in V(T_S)\backslash L(T_S)$ is a lost subtree in $c(\R_i(e_1))$, then there are lost subtrees (in $c(\R_i(e_1))$) $v_1,\ldots, v_s$, $\rho(v_0)<\rho(v_1)<\ldots <\rho(v_s)$, $v_i$ overlaps with $v_{i+1}$ $(i=0,\ldots, s-1)$ and $v_s=t_i$. 
	
	Let $v\in V(T_S)\backslash L(T_S)$ be a lost subtree. Let us prove that $v$ is a free tree (in $c(\R_{i-1}(e_1))$, $c(\R_i(e_1))$, and $c(\R)$). From $v\in V(T_S)\backslash L(T_S)$ we have $v=v_0<v_1<\ldots,v_{s-1}<v_s=t_i$ and $v_{i-1}$ overlaps $v_i$. Since $v_{s-1}$ overlaps $t_i$ (in $c(\R_i(e_1))$) and $t_i$ is the same in both $c(\R_i(e_1))$ and $c(\R)$, we have that $v_{s-1}$ overlaps $t_i$ in $c(\R)$, hence $v_{s-1}$ is a free tree in $c(\R)$, {\em i.e.} $v_{s-1}\in \{t_1,\ldots, t_{i-1}\}$. Applying the same argument on $v_{s-1}$, we get $v_{s-2}\in \{t_1,\ldots,t_{i-1}\}$. Proceeding in this manner, we have $v\in \{t_1,\ldots,t_{i-1}\}$, hence $v$ is a free tree.      
	
	Let $f_1,\ldots, f_r$ be the children of leaf edges of $T_S$. From the maximality of $T_S$, we have there is no lost subtree in $c(\R_{i-1}(e_1))$ nor in $c(\R_i(e_1))$ that expands over $f_j$,  $(j=1,\ldots, r)$. All non-free losses from $S(e_1)$ are contained in $S(f_j)$,  $(j=1,\ldots, r)$. This holds for both $c(\R_{i-1})$ and $c(\R_i)$. Therefore the structure of the lost subtrees in $\R_{i-1}(f_j)$ can be identical to the structure of the lost subtrees in $\R_{i}(f_j)$,  $(j=1,\ldots, r)$, and thus obtaining that a completion of $\R_{i-1}(e_1)$ has the same weight as an extension of  $\R_{i}(e_1)$, a contradiction.     
	
	Hence the procedure \textproc{ExtendLossIntoFreeTree} can give us $t_i$, ($i=1,\ldots,k$).  \qed   
\end{proof}

It is proved in Section \ref{sec:algorithm}, in a more general framework, that these procedures indeed compute a completion, and hence, if the input reconciliation is the LCA reconciliation, it computes an optimal reconciliation.

	\section{Zero-flow reconciliations and the space of all optimal reconciliations}\label{sec:all}

Here we introduce zero-flow reconciliations and use them as a hinge to find all optimal reconciliations. Zero-flow (ZF) reconciliations are a subspace of optimal reconciliations and they contain LCA reconciliations, but these inclusions are strict: all sets are distinct. We first show how to find any ZF reconciliation, up to completion, from an LCA reconciliation. Then by a different procedure we show how to access the whole space of optimal reconciliations, up to completion, from a ZF reconciliation. Finally, as these reductions work up to completion, we show how to navigate in all completions for a given reconciliation.

Let $e=(s,p(s))$ be an edge of $S$ and $\R$ a reconciliation. We note
$X(e,\R)=\{d\in V(G)\mid \rho_{lca}(d)\le s, \rho(d)\ge p(s)\}$ the set of nodes (duplications or conversions) which are assigned under $s$ in the LCA reconciliation and above $p(s)$ in $\R$.

\begin{definition}
An optimal reconciliation $\R$ is said to be a {\em zero-flow} (ZF) reconciliation if for all $s$ internal node of $S$ with children edges $e_1$ and $e_2$, $F(e_1,\R)<0 \implies X(e_1,\R) = X(e_2,\R)=\emptyset$.
\end{definition}

In other words, an optimal reconciliation is ZF if all duplications assigned to or above a node $s$, when strictly below in the LCA, verify that the flow the children edges of $s$ is non negative.
By definition LCA reconciliations are ZF ($X(e,\R_{lca})=\emptyset$ for all $e$). But we will see that the converse is not true. Similarly ZF reconciliations are optimal by definition but some optimal reconciliations are not ZF.

\subsection{Computing ZF reconciliations by duplication raising}

Duplication raising consists in changing the position of a duplication from its position in a minimal reconciliation to an upper position in the species tree. It is a concept that was previously used to explore DL reconciliations \cite{Chauve2008}.

\begin{definition}[Node raising]
	Let $\R=(G,G',S,\phi,\rho, \delta)$ be a minimal reconciliation and  $x\in V(G)$. We say that reconciliation $\R'=(G,G'',S,\phi,\rho', \delta')$ is obtained from $\R$ by {\em raising} node $x$ if $\R'$ is a minimal reconciliation such that $\rho(x')=\rho'(x')$, $\forall x'\in V(G)\backslash \{x\}$ and $\rho'(x)=p(\rho(x))$.
\end{definition}

\begin{figure}[h]
	\centering 
	\includegraphics{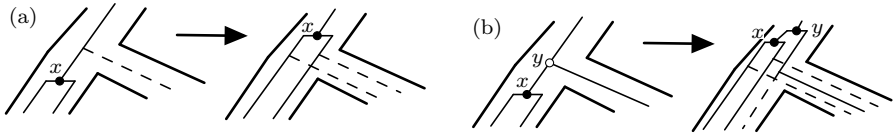}
	\caption{Duplication raising. Given duplication $x$ (a) No speciation. After raising $x$, a new loss is created. An optimal solution can be generated by this operation. (b) We have $y=p_G(x)$ and $\rho(y)=p(\rho(x))$. After raising $x$, $y$ becomes duplication and three new losses are generated. This cannot be optimal}
	\label{fig:duplication_raising}
\end{figure}

Depending on the assignment and event status of the parent node of $x$, raising $x$ has different effects. If $p_{G}(x)$ is a speciation (see Figure \ref{fig:duplication_raising}) and $\rho(p_{G}(x))=p(\rho(x))$, after raising $x$, $p_G(x)$ becomes a duplication and three new losses are generated. This cannot lead to an optimal solution because of the additional duplication (Theorem \ref{th:duplic_const}). If $\rho(p_G(x))>p(\rho(x))$ or $p_G(x)$ is a duplication, after raising $x$, only one  additional loss is generated. This condition, which is necessary to yield an optimal solution, is formalized as follows.
\begin{equation}
x\in \Delta(\R) \wedge  \bigg(p(\rho(x))<\rho(p_G(x)) \vee \Big(p(\rho(x))=\rho(p_G(x)) \wedge p_G(x)\in \Delta(\R) \Big) \bigg) \label{eq:star}
\end{equation}

The next lemma states that raising a duplication cannot decrease the weight of a completion. The proof of the lemma also describes how to lower a duplication. This procedure will be important later in some proofs. 
\begin{lemma}\label{lem:duplRaising_is_expensive} 
	Let $\R$ be a minimal reconciliation, $\R_1$ is a minimal reconciliation obtained from $\R$ by raising a duplication. Then $\omega(c(\R))\le \omega(c(\R_1))$. 
\end{lemma}
\begin{proof}
	Let $x$ be the raised duplication, $e_1,e_2\in E(S)$ are siblings, $e$ is their parent, $x$ is assigned to $e_1$ in $\R$ and to $e$ in $\R_1$. 
	
	Let $T$ be the lost subtree such that $root(T)$ is a child of $x$ in $c(\R_1)$ and $T$ is expanded over $e_2$. Observe two cases (see Figure \ref{fig:duplication_lowering}).
\begin{figure}[h]
    \centering 
    \includegraphics{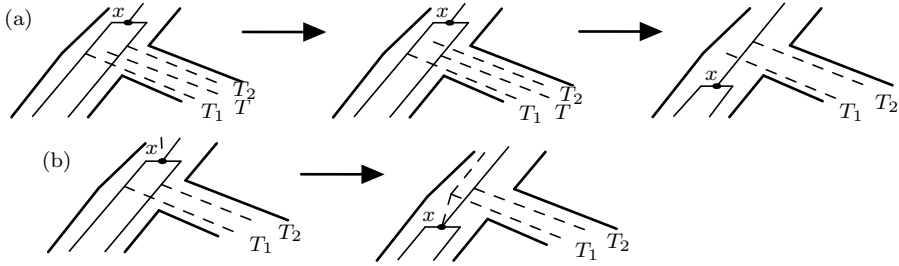}
	\caption{Lowering duplication. (a) Lowering duplication (that is not a conversion). If there is non-free subtree $T$, not associated with $x$, we can make it associated by rearranging the roots. By lowering $x$ we can delete one lost subtree ($T$), and if it is non-free, then we get a cheaper reconciliation. (b) Lowering a conversion. Loss assigned to $x$ can be extended to one of the lost subtrees on the right side. We get a reconciliation of the same weight}
    \label{fig:duplication_lowering}
\end{figure}

	Case 1, $x\notin \Delta'(c(\R_1))$. Start with $c(\R_1)$, place $x$ back to $e_1$, and remove $T$. We get an extension of $\R$ with a cost at most the one of $c(\R_1)$, {\em i.e.} $\omega(c(\R))\le \omega(c(\R_1))$.
	
	Case 2, $x\in\Delta'(c(\R_1))$. Let $l$ be a loss assigned to $x$ in $c(\R_1)$. Start with $c(\R_1)$, place $x$ back to $e_1$, extend $l$, so that in $e_1$ is paired with $x$ (staying free loss), and in $e_2$ is connected to $T$. In this way, we get an extension of $\R$ of the same weight as $c(\R_1)$, i.e. $\omega(c(\R))\le \omega(c(\R_1))$. \qed
\end{proof}

As a consequence of Lemma \ref{lem:duplRaising_is_expensive}, no optimal reconciliation can be obtained by raising a duplication from a reconciliation that has no optimal completion. We will now see in which conditions a duplication raising of a reconciliation with an optimal completion can lead to another reconciliation with optimal completion.

The next lemma states when raising a duplication does not increase the weight of a reconciliation. 
\begin{lemma}\label{lem:duplRaising_F>0}
	Let $\R$ be a minimal reconciliation, and $e_1,e_2\in E(S)$ the children of edge $e$. If $x\in\Delta(\R)$ assigned to $e_1$ satisfies condition $\eqref{eq:star}$, $\R_1$ is a minimal reconciliation obtained by raising $x$, $F(e_1,\R)>0$ and $F(e_2,\R)>0$, then $\omega(c(\R_1))=\omega(c(\R))$.
\end{lemma}
\begin{proof}
	 First, construct an extension of $\R_1$, by using $c(\R)$. By raising $x$, we generate one new loss in $e_2$. Since $F(e_2,\R)>0$, we have $\omega(c(\R(e_2)))=\omega(c(\R_1(e_2)))$, {\em i.e.} the loss generated by the duplication raising can become a free loss. 
     
     Let $x\in\Delta'(c(\R))$ and assigned to $l\in \Lambda'(c(\R))$. If $l$ is non-extended (in $c(\R)$)  and since $F(e_1,\R)>0$, we have that $l$ can be assigned to some other duplication in $e_1$ or extend over children of $e_1$ and become free. If $l$ is part of a lost subtree $T_l$ in $c(\R)$, then by raising $x$, we ca also raise $l$, remove subtree of $T_l$ expanding over $e_2$, leave $l$ assigned to $x$.  
     
     Thus we obtain an extension of $\R_1$, not heavier than $c(\R)$, i.e. $\omega(c(\R_1))\le \omega(c(\R))$. From Lemma \ref{lem:duplRaising_is_expensive}, we have $\omega(c(\R))\le \omega(c(\R_1))$, hence $\omega(c(\R_1))= \omega(c(\R))$.   \qed  
\end{proof}

The next lemma follows directly from Lemma \ref{lem:duplRaising_F>0}.
\begin{lemma}\label{lem:duplRaising_in_optimal}
	Under the hypotheses of Lemma \ref{lem:duplRaising_F>0}, if completions of $\R$ are optimal, then completions of $\R_1$ are optimal.
\end{lemma}

Algorithms \ref{RaiseSeveralDuplications}, \ref{dPosition}, and \ref{PossiblePositions} describe how to generate a reconciliation which does not change the score of completions by raising duplications.

  \begin{algorithm}[h]
  	\caption{Raises duplications}\label{RaiseSeveralDuplications}
  	\begin{algorithmic}[1]
  		\Procedure{RaiseSeveralDuplications}{$d,\R$} 
  		
  		 \For{$d \in \Delta(\R_{lca})$ from top to bottom of $V(G)$} 
  		  \State $\R \leftarrow$ \textproc{RaiseDuplication}($d,\R_{lca}$)
  		 \EndFor
  		\EndProcedure
  	 \end{algorithmic}	
  \end{algorithm}

 \begin{algorithm}
 	\caption{Raises duplication (respecting $F>0$)} \label{dPosition}
 	\begin{algorithmic}[1]
     \Procedure{RaiseDuplication}{$d,\R$}
     \State L $\leftarrow$ \textproc{PossiblePositions}($d,\R$) 
     \State k = \textbf{random}(0,$|L|-1$)  
     \State $\rho(d)=L[k]$  
     \State \textproc{GenerateNewLosses()}
     \EndProcedure
    \end{algorithmic}
 \end{algorithm}

Procedure \textproc{GenerateNewLosses} adds lost subtrees to that the new $\rho$ after raising a duplication is consistent with $S$.

 \begin{algorithm} 
 	\caption{Possible new positions for a duplication.} \label{PossiblePositions}
 	\begin{algorithmic}[1]
 	 \Procedure{PossiblePositions}{$d,\R$}
 	   \State $s=\rho(d)$ 
 	   \State $e=(s,p(s))$ 
 	   \State $e'$ - sibling of $e$
 	   \State $L \leftarrow \{s\}$
 	   \While{$F(e)>0$ and $F(e')>0$ and ($\rho(p_G(d))>p(s)$ or ($\rho(p_G(d))== p(s)$ and $p_G(d)\in \Delta$ ))} 
 	    \State $s=p(s)$
 	    \State $e=(s,p(s))$
 	    \State $e'$ - sibling of $e$
 	    \State $L\leftarrow L+ \{s\}$
 	    
 	   \EndWhile
      
     \EndProcedure
    \end{algorithmic}
 \end{algorithm}

The two next statements demonstrate that, up to completion, all the ZF reconciliations are reached by applying Algorithm \textproc{RaiseDuplication} on a LCA reconciliation.

\begin{lemma}\label{lem:Changeposition}
	Completions of $\R$, an output of \textproc{RaiseDuplication} when the input is the LCA reconciliation, are optimal. 
\end{lemma}
\begin{proof}
	Completion of LCA reconciliation is an optimal (Theorem \ref{th:LCAopt}), raised duplications satisfy conditions of Lemma \ref{lem:duplRaising_in_optimal}, and by this Lemma every time a duplication is raised we get that $c(\R)$ is an optimal reconciliation.   \qed
\end{proof}

\begin{lemma}\label{lem:RaiseDuplication}
	Let $\R'$ be a minimal reconciliation such that $c(\R')$ is a ZF reconciliation. Then $\R'$ is a possible output of \textproc{RaiseDuplication}. 
\end{lemma}
\begin{proof}
	Since $c(\R')$ is an optimal reconciliation, $\R'$ is obtained from LCA by raising duplications that satisfy condition \eqref{eq:star}. By raising a duplication, value of $F(e)$ cannot increase. Let $e_1,e_2\in E(S)$ be siblings, $e$ their parent, $x$ a duplication assigned to $e_1$. Let us raise $x$ to $e$. If before raising $F(e_1)\le 0$ or $F(e_2)\le 0$, then after raising $F(e_1)< 0$ or $F(e_2)< 0$, $X(e_1)\ne \emptyset$, and $X(e_2)\ne \emptyset$, a contradiction. Hence $F(e_1)> 0$ and $F(e_2)> 0$. 
	
	Thus all conditions, for raising a duplication, of the procedure \textproc{RaiseDuplication} are satisfied, hence $\R'$ is a possible output.    \qed
\end{proof}


\subsection{Reduction of optimal reconciliations to ZF reconciliations}

Lemma \ref{lem:RaiseDuplication} states that up to completion, we can generate all ZF reconciliation from LCA reconciliations. We now show how to generate all reconciliations from ZF reconciliations. This is done by conversion raising. Next lemma proves that only conversions are concerned by optimal non ZF reconciliations.

\begin{lemma}\label{lem:X_in_Delta}
	Let $\R$ be an optimal reconciliation, $e_1=(s_1,s),e_2=(s_2,s)\in E(S)$. If $F(e_1,\R)<0$, then $X(e_1,\R)$ and $X(e_2,\R)$ are only conversions. 
\end{lemma}
\begin{proof} 
Assume the opposite, let $x\in X(e_1,\R)$ and $x$ is not a conversion. Put back (lower) all elements of $X(e_1,\R)$ to $e_1$. The process is performed as in the proof of Lemma \ref{lem:duplRaising_is_expensive} (Figure \ref{fig:duplication_lowering}). If we lower a conversion, the weight of a reconciliation is not changed, as well as $F(e_1)$. If we lower a duplication, then $F(e_1)$ is increased by 1 and the cost of a completion is decreased by one (Lemmas \ref{lem:FLessThanZero}, \ref{lem:lossRemoved} and the comment after), which is a contradiction with the optimality of $\R$. Therefore, $X(e_1,\R)$ does not contain a duplication that is not a conversion.   

Similar arguments apply to $X(e_2,\R)$.    \qed  
\end{proof}

\begin{lemma}\label{lem:RaiseConversions}
	Procedure \textproc{RaiseConversions} does not change the weight of a reconciliation.
\end{lemma}
\begin{proof}
	 Let $d$ be a raised  conversion, and $T_i$ is a lost subtree whose leaf is assigned to $d$. By raising $d$, we do not create an extra losses, but use existing subtree of $T_i$ and reattach it under $d$ (see Figure \ref{fig:duplication_lowering} $(ii)$ in the opposite direction and Lemma \ref{lem:duplRaising_is_expensive}, Case 2). The loss that was assigned to $d$ is removed, and newly created loss is assigned to $d$ at a new position. In this way we do not change the number of non-free losses, and the number of duplications/conversions, i.e. the weight of the reconciliation is not changed.	 \qed
\end{proof}

\begin{lemma}\label{lem:L>0_X=Y=0}
	Let $\R$ be an optimal reconciliation. We can obtain a ZF reconciliation by lowering some conversions.  
\end{lemma}
\begin{proof}
	For all $e \in E(S)$, if $F(e)<0$, take all elements from $X(e)$ and $X(e')$, where $e'$ is the sibling of $e$, and lower them to $e$ and $e'$. In this way we get $X(e)=X(e')=\emptyset$. Since these elements are conversions (Lemma \ref{lem:X_in_Delta}) lower them as described in Lemma \ref{lem:duplRaising_is_expensive}, Case 2. 
    
	In this way we obtain a ZF reconciliation of the same weight as $\R$.  \qed
\end{proof}

In consequence it is possible to reach any optimal reconciliation by an algorithm which explores first ZF reconciliations and raises some conversions as in Algorithm \ref{raise_some}.

 \begin{algorithm} 
	\caption{raises some conversions}\label{raise_some}
	\begin{algorithmic}[1]
	\Procedure{RaiseConversions}{$\R$} 
        \State By convention let $e_1(d)$ denote the edge to which $d$ is assigned, and $e_2(d)$ its sibling in $S$.
		\State Let $C=\{d\mid d\in \Delta', F(e_1(d))< 0 \text{ or } F(e_2(d))< 0\}$
		\State Let $T_d$ be used to denote the lost subtree with a leaf paired with $d$ by $\delta$.
		\While{$C\ne \emptyset$}
			\State $d\in C$ - \textbf{random}
			\State \textproc{RaiseOneConversion$(d,\R,T_d)$}
			\State $C=C\backslash\{d\}$ 
		\EndWhile
		
	\EndProcedure
	\end{algorithmic}
\end{algorithm}

 \begin{algorithm}
 	\caption{raises one conversion}\label{raise_one}
 	\begin{algorithmic}[1]
 	\Procedure{RaiseOneConversion}{$d,\R,T_d$} 
    \State Let $s$ be a random element of $V(S)$ satisfying
    \State (i) $s\ge p(\rho(d))$ 
    \State (ii) $s\le min(\rho(root(T_d)),\rho(p_G(d)))$  
    \State (iii) if $p_G(d)\in \Sigma$ then $s\ne \rho(p_G(d))$
    \State Note $\rho(d)=s_0<s_1<\ldots<s_k=s$
    \State $T_d^j$ - subtree of $T_d$, $\rho(root(T_d^j))=s_j$, $j=\overline{1,k}$
    \State assign $d$ to random $s_i$
    \State node (leaf) of $T_d$, assigned to $s_i$, pair with $d$ (and $d$ stays conversion)  
    \State root of every tree $T_d^j$ position in $G'$, under $d$, at an appropriate position
	\EndProcedure
\end{algorithmic}
\end{algorithm}

\subsection{Finding all completions}

All previous results are valid up to completions. It means that we have an algorithm which is able to detect all duplications that can be conversions in one optimal solution for example. However we don't know all the possibilities by which it is converted. For that we need to enumerate all possible completions.
The algorithm can be described by three procedures, as written in Algorithm \ref{extend_losses}.

 \begin{algorithm}
 	\caption{finds a random completion}\label{extend_losses}
 	\begin{algorithmic}[1]
 	\Procedure{AllCompletions}{$\R$}
 		  \State \textproc{OneCompletion}($\R$) \
 		  \State \textproc{ExtendLossesIntoNonFreeTrees}($\R$)
          \State \textproc{Switch}($\R_c$) 
 	\EndProcedure
 	\end{algorithmic}
 \end{algorithm}

One procedure is to generate a completion by extending losses into free trees, which is described in Section \ref{sec:flow}. In order to generate the full diversity of possible reconciliations, there are two others described here, which consist in extending losses into non free lost subtrees, and switch between subtrees. The first one is described in Algorithms~\ref{ExtendNonFreeLosses} and \ref{ExtendOneNonFreeLoss}.
In Algorithm \ref{ExtendOneNonFreeLoss} a loss is extended over two edges, one with positive $F$-value (say edge $e_1$), and the other with non-positive $F$-value (say edge $e_2$). The part (of the lost subtree) extended over $e_1$ is further extended as a free loss, while the part extended over $e_2$ is further (recursively) extended as a non-free loss.

 \begin{algorithm}
	\caption{randomly extends losses into non-free trees} \label{ExtendNonFreeLosses}
	\begin{algorithmic}[1]
	\Procedure{ExtendLossesIntoNonFreeTrees}{$\R$}
		  \State $\Sigma_1$ is the set of all non-free, non-extended losses in $\R$
		  \For{all $l\in \Sigma_1$}
		    \State \textproc{ExtendOneLossIntoNonFreeTree}$(\R,l)$
		  \EndFor
	\EndProcedure
	\end{algorithmic}
\end{algorithm}

\begin{algorithm}
	\caption{randomly extends losses into non-free trees}\label{ExtendOneNonFreeLoss}
	\begin{algorithmic}[1]
		\Procedure{ExtendOneLossIntoNonFreeTree}{$\R,l$}
           \State $l$ is assigned to $e=(s,p(s))$
           \State $e_1,e_2$ are children of $e$ and $F(e_1)\ge F(e_2)$
           \State Randomly choose between "extend" or not.
           \If{$F(e_1)>0$ and $F(e_2)\le 0$ and "extend" has been chosen}
              \State extend $l$ over $e_1,e_2$
              \State $l_1,l_2$ are new losses assigned to $e_1,e_2$ and $l$ is their parent
              \State \textproc{ExtendOneLossIntoFreeTree}$(\R,l_1)$
              \State \textproc{ExtendOneLossIntoNonFreeTree}$(\R,l_2)$
           \EndIf         
		\EndProcedure
	\end{algorithmic}
\end{algorithm}

\begin{lemma}\label{lem:nonFreeTree}
	Let $l$ be a non-free loss in a reconciliation $\R$. Then procedure \textproc{ExtendOneLossIntoNonFreeTree}$(\R,l)$ extends loss $l$ into a non-free tree. 
\end{lemma}
\begin{proof}
	If $l$ is not extended, since it is not assigned to a duplication (conversion) we will assume that it is extended into a non-free tree (with one edge).  
	
	Let $l$ be assigned to the edge $e$, and $e_1,e_2$ are its children. We will use mathematical induction on $e$.
	
	Let $e$ be a leaf edge. Then $e_1=NULL,e_2=NULL$ and $F(e_1)=F(e_2)=0$. In this case, the {\em if} condition is not satisfied, and therefore $l$ is not extended.
	
	Assume that $e$ is not a leaf edge. If the {\em if} condition is not satisfied, then $l$ is not extended, {\em i.e.} it is extended into a non-free tree with one edge. If the {\em if} condition is satisfied, then $F(e_1)>0$ and $F(e_2)\le 0$, and $l$ is extended into $l_1,l_2$. Then \textproc{ExtendOneLossIntoFreeTree}$(\R,l_1)$ extends $l_1$ into a free tree (Lemma \ref{lem:ExtendLossIntoFreeTree}), and \textproc{ExtendOneLossIntoNonFreeTree}$(\R,l_2)$ extends $l_2$ into a non-free tree (inductive hypothesis). Hence $l$ is extended into a non-free tree.  \qed   
\end{proof}

The next lemma is a consequence of Lemma \ref{lem:nonFreeTree}
\begin{lemma}\label{lem:ExtendNonFreeLosses}
	Procedure \textproc{ExtendLossesIntoNonFreeTrees} does not change the weight of a reconciliation.
\end{lemma}

\begin{lemma}\label{lem:CompletionOutput}
	Let $\R$ be a reconciliation with non-extended losses, $t_i$ ($i=1\ldots k$) and $t'_j$ ($j=1\ldots m$) are free and non-free lost subtrees of $c(\R)$ such that $t'_j\ge t_i$ whenever $t_{i}$ and $t'_{j}$ overlap. Then $c(\R)$ is a possible output of series of procedures \textproc{OneCompletion}$(\R)$,  \textproc{ExtendLossesIntoNonFreeTrees}$(\R)$.   
\end{lemma}
\begin{proof}
	Let $\R_0=\R$, $\R_i$ is obtained from $\R_{i-1}$ by extending corresponding loss to the tree $t_i$ ($i=1,\ldots, k$), $\R'_0=\R_k$, $\R'_j$ is obtained from $\R'_{j-1}$ by extending corresponding loss to the tree $t'_j$ ($j=1,\ldots, m$). Hence $\R'_m=c(\R)$.
	
	The procedure \textproc{OneCompletion} can give us $t_i$, ($i=1,\ldots,k$) (Lemma \ref{lem:CompletionOutput_1}). Now we will prove that \textproc{ExtendLossesIntoNonFreeTrees} can give us $t'_j$, ($j=1,\ldots, m$). 
	
	Assume that $t_i$,  ($i=1,\ldots,k$), $t'_1,\ldots, t'_{j-1}$ ($j\ge 1$) are added. Let us prove that \textproc{ExtendLossesIntoNonFreeTrees} can add $t'_j$. Let  $e_1,e_2\in E(S)$, $e=(s,p(s))$ is their parent, and $\rho(l'_j)=s$, where $l'_j$ extends into $t'_j$. If $F(e,\R'_{j-1})>0$, then $l'_j$ can be free, thus obtaining a cheaper reconciliation than $c(\R)$, a contradiction, so $F(e,\R'_{j-1})\le 0$.  
	
	Let $e'_1,e'_2\in E(\rho(t'_j))$ be siblings, $e'$ their parent, and $F(e_1',\R'_{i-1})\ge F(e_2',\R'_{i-1})$. Subtree $t'_j$ expands over $e'_1,e'_2$ and not necessarily originating at $e'$. Observe two cases.
	
	Case 1, $F(e',\R'_{j-1})\le 0$. If $F(e'_1,\R'_{j-1})\le 0$ (and $F(e'_2,\R'_{j-1}) \le 0$), then by pruning  $t'_j$ both $e'_1$ and $e'_2$ don't gain a loss, so the cost of reconciliations $c(\R'_{j-1}(e'_1))$ and $c(\R'_{j-1})(e'_2)$ will not rise in $\R'_j$, but $\R'_j$ gain one non-free loss (pruned $t'_j$). Hence we gain a cheaper reconciliation, a contradiction. 
 	
	Assume $F(e'_1,\R'_{j-1})>0$ and $F(e'_2,\R'_{j-1}) > 0$. Since $F(e',\R'_{j-1})\le 0$, there is a loss $l$ assigned to $e'$ that is non-free (in $\R'_{j-1}$). Then we can extend $l$ over $e'_1,e'_2$ so it becomes free, and prune $t'_j$ to a single edge ($t'_j$ stays non-free). Hence obtaining a cheaper reconciliation than $c(\R)$, a contradiction. 
	
	Case 2, $F(e',\R'_{j-1})>0$. If $F(e'_2,\R'_{j-1})\le 0$, then $e'$ has a duplication that is not a conversion.  At least one of the subtrees of $t'_j$ expanding over $e'_1,e'_2$ is a free tree. Assume that it is the one expanding over $e'_1$. Next, we can prune subtree of $t'_j$ so that $t'_j$ has a leaf assigned to $e'$ and to the duplication, thus becoming a free loss. Since $F(e'_2,\R'_{j-1})\le 0$ there is one non-free loss in $\R'_{j-1}(e'_2)$ that can become free, thanks to the fact that $t'_j$ does not expand over $e'_1$ anymore. Making this loss free enable us to obtain a cheaper reconciliation than $c(\R)$, a contradiction.  
   
	From the Cases 1 and 2, we have that if $F(e',\R'_{j-1})\le 0$, then $F(e'_1,\R'_{j-1})>0$, $F(e'_2,\R'_{j-1})\le 0$, and if $F(e',\R'_{j-1})> 0$, then $F(e'_1,\R'_{j-1})>0$, $F(e'_2,\R'_{j-1})> 0$. Hence conditions along $\rho(t'_j)$ of  \textproc{ExtendLossesIntoNonFreeTrees} are satisfied, and therefore $t'_j$ can be obtained by this procedure.      \qed
\end{proof}

To obtain all possible lost subtrees in an optimal reconciliation, we need to introduce an operation that exchanges parts of the lost subtrees. 
Notice that a lost subtree with more than one non-free leaf cannot appear in an optimal reconciliation. 

\begin{definition}[Switch operation on a binary rooted trees]
	Let $T_0$ and $T_1$ be binary rooted trees and $t_i\in V(T_i)\backslash\{root(T_i)\}$ $(i=0,1)$. A {\em switch} operation on $T_0$ and $T_1$ around $t_0$ and $t_1$ creates new trees by separating subtrees $T_i(t_i)$ from $T_i$ and joining them with $p(t_{1-i})\in T_{1-i}$ $(i=0,1)$. 
\end{definition}

\begin{subfigures}
\begin{figure}
	\centering 
	\includegraphics{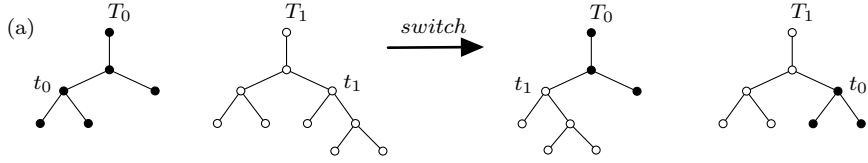}
	\caption{Switch operation between binary trees.  (a) Switch between $T_0$ and $T_1$ around $t_0$ and $t_1$}
    \label{fig:switcha}
\end{figure}
\begin{figure}
	\centering 
	\includegraphics{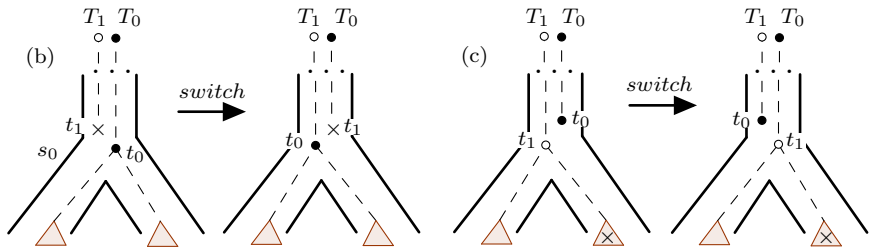}
	\caption{Switch operation between reconciliations. (b-c) Switch on a reconciliation. Exactly one lost subtree receives a (nontrivial) subtree from the other lost subtree. A subtree with a non-free loss has to be involved in a switch operation. An empty triangle denotes a free subtree, while a triangle with $x$ denotes a non-free subtree}
	\label{fig:switchb}
\end{figure}
\end{subfigures}

\begin{figure}
	\centering
	\includegraphics{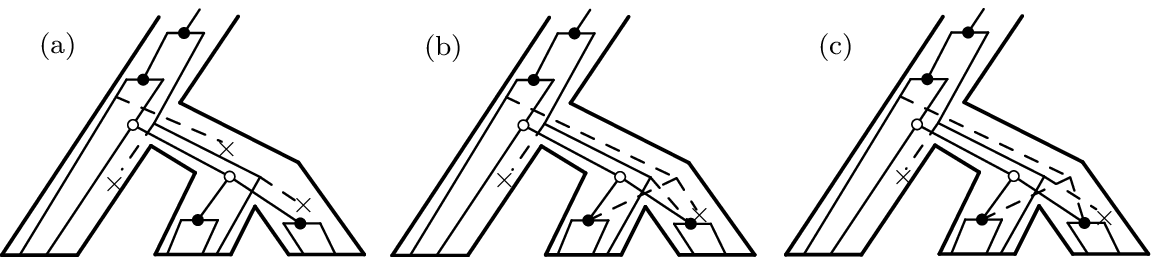}
	\caption{Example for necessity of switch operation. (a) Minimal reconciliation. (b) The completion. We have one free and two non-free trees. (c) A completion obtained by switch operation. Note that this completion is not obtainable by standard extension into free and non-free trees}
\end{figure}

\begin{definition}[Switch operation on a reconciliation ]\label{def:switch}
	Let $\R$ be a reconciliation, $T_0$ and $T_1$ free and non-free lost subtrees, $l\in L(T_1)$ is a non-free loss, $p$ is a path in $S$ from $\rho(l)$ to $\rho(root(T_1))$. Assume there exists a minimal element $s_0\in \{s\mid s\in V(p)\cap V(\rho(T_0)) \} \backslash \{\rho(root(T_0)), \rho(root(T_1))\}$, and $t_i\in V(T_i)$ such that $\rho(t_i)=s_0$ $(i=0,1)$. By {\em switch} operation on $T_0$ and $T_1$ we mean a switch operation on the binary trees $T_0$ and $T_1$ around $t_0$ and $t_1$.  
\end{definition} 

Switch operation on a reconciliation is defined only for one free and one non-free lost subtree, and is possible only if trees $T_0$ and $T_1$ {\em overlap}, {\em i.e.} if $\rho(v_0) \in \rho(T_1)$ or $\rho(v_1) \in \rho(T_0)$, where $v_i=root(T_i)$, $(i=0,1)$.
In the case $\rho(v_0) \in \rho(T_1)$, it must be $\rho(l)<\rho(v_0)$, where $l$ is a non-free leaf of $T_1$. In these cases we say that $T_0$ and $T_1$ are {\em switchable}. We have that either $T_0$ gives a (non-trivial) subtree to $T_1$, or $T_1$ gives a (non-trivial) subtree to $T_0$, but both cannot happen.  

When we apply a switch operation two times on the same trees, around the same nodes, we obtain starting trees, i.e. switching is self-inverse operation. After switch operation, involved trees still overlap. 

For simplicity of notation, we introduce some conventions. We write \emph{tree} instead of \emph{lost subtree}. We will identify a tree with its root, {\em i.e.} instead of writing \emph{a tree with the root $v$}, we will use \emph{a tree $v$}. We do this because, when switching, trees are changed, but the roots are not. When we write $v_0<v_1$, we mean $\rho(v_0)<\rho(v_1)$. Number of non-free leaves in a tree $v$ is denoted by $\omega(v)$, thus $\omega(v)=0$ means that $v$ is a free lost subtree, and $\omega(v)=1$ means that $v$ is a non-free lost subtree.   

If we will apply a switch operation on switchable trees $v_0$, $v_1$ such that $\omega(v_1)=1$ and $\omega(v_0)=0$, we say that $v_1$ carries over a (non-free) loss to $v_0$.

The next lemma is obvious.  
\begin{lemma}\label{lem:switchIsConst}
	Switch operation does not change the weight of a reconciliation.
\end{lemma}

The next lemma tells us how to, from an arbitrary reconciliation, obtain a reconciliation with more convenient structure of lost subtrees.     
\begin{lemma}\label{lem:arrangedLostSubtrees}
	Let $\R$ be a reconciliation. Then there exists a reconciliation $\R_1$ such that if $v_0$ and $v_1$  are free and non-free overlapping trees in $\R_1$, then $v_0\le v_1$ and $\omega(\R)=\omega(\R_1)$.  
\end{lemma}
\begin{proof}
	Let $V_{lost}=\{v\mid v \text{ is a lost subtree}\}$. Take  $v_0\in V_{lost}$ such that $\omega(v_0)=1$, and $v_1\in V_{lost}$ such that $\omega(v_1)=0$,  $v_0<v_1$, and $v_0$ is overlapping with $v_1$. By switching $v_0$ and $v_1$ we get $\omega(v_0)=0$, $\omega(v_1)=1$ and $v_0<v_1$. Repeat the process as long as there are trees $v_0,v_1$ as described. We need to prove that this algorithm ends.
	
	Let $d(V_{lost})$ be the total distance of all non-free $v\in V_{lost}$ from $root(S)$. Hence $d(V_{lost})$ is a non-negative integer. Every time, when switching is applied,  $d(V_{lost})$ decreases, hence the algorithm must stop, because  $d(V_{lost})$ cannot decrease indefinitely.   
	
	Switch operation does not change the weight of a reconciliation (Lemma \ref{lem:switchIsConst}).  \qed         
\end{proof}

\begin{algorithm}
	\caption{applies switch operation on lost subtrees }\label{switch}
	\begin{algorithmic}[1]
	\Procedure{Switch}{$\R$} 
	    \State $\mathfrak{T}'$ - the set of all non-free lost subtrees in $\R$
	    \State $\mathfrak{T}_{v'}$ - the set of all free lost subtrees, less than $v'$, switchable with $v'$
        \While{$\mathfrak{T}' \ne \emptyset$}
          \State $v'\in \mathfrak{T}'$ - \textbf{random}
          \State $v\in \mathfrak{T}_{v'}\cup \{NULL\}$
          
          \If{$v$==NULL}
            \State $\mathfrak{T}'=\mathfrak{T}'\backslash \{v'\}$
            \State \textbf{continue} while loop
          \EndIf
          
          \State \textproc{SwitchSubtrees}($v,v'$)
          \State $\mathfrak{T}'=(\mathfrak{T}'\backslash \{v'\})\cup\{v\}$ 
        \EndWhile
    \EndProcedure
    \end{algorithmic}
\end{algorithm}

Procedure \textproc{SwitchSubtrees} is described in Definition \ref{def:switch}.

	\section{The algorithm}\label{sec:algorithm}

In this section, we prove that the algorithm returns an optimal reconciliation, and any optimal reconciliation can be an output of the algorithm. We also prove the remaining lemmas.

All elements are ready to write the main algorithm that generates a random optimal solution.

  Algorithm \ref{randR} gives the main procedure.

  \begin{algorithm}[h]
  	\caption{Random reconciliation}\label{randR}
  	\begin{algorithmic}[1]
  		\Procedure{RandR}{$S,G,\phi$} 
  		 \State Let $\R_{lca}$ be the LCA reconciliation

  		  \State $\R \leftarrow$ \textproc{RaiseSeveralDuplications}($d,\R_{lca}$)
  		
  		 \State $\R_c \leftarrow$ \textproc{AllCompletions}($\R$)
  		 \State Return \textproc{RaiseConversions}($\R_c$) 
  		\EndProcedure
  	 \end{algorithmic}	
  \end{algorithm}

Now we prove a lemma stated earlier.

\begin{proof}[Proof of Lemma \ref{lem:FLessThanZero}]
Let $l$ be the number of assigned losses to $e$ in $\R_1$, $\R$ is the (multiple) reconciliation obtained from $\R_1$ by removing all ($l$) losses from $e$, $k'$ is from the definition of flow. Then $F(e,\R_1)=k'-l$. Therefore, the maximum number of extra losses that we can assign to $e$ in $\R$, without completion cost change, is $k'$ and $k'\le l$.  

It is obvious that $\Delta(\R)=\Delta(\R_1)=\Delta(\R_2)$. Also $\omega(c(\R))<\omega(c(\R_2))$. 

 We have that $\omega(c(\R_2))=\omega(c(\R_1))+1$ or $\omega(c(\R_2))=\omega(c(\R_1))$. Assume that $\omega(c(\R_2))=\omega(c(\R_1))$.

 Observe $c(\R_2)$. Let $t_1,\ldots,t_l,t_{l+1}$ be the lost subtrees with the roots assigned to $p(s)$ (and expanding over $e$). If any of these subtrees are non-free in $c(\R_2)$ then by removing it we get an extension of $\R_1$ that has strictly less weight than $\omega(c(\R_2))=\omega(c(\R_1))$, a contradiction. Therefore all subtrees $t_1,\ldots,t_l,t_{l+1}$ are free in $c(\R_2)$.
 
Let us prove that there is at least one non-free subtree in $c(\R_2)$. Assume the opposite, {\em i.e.} all lost subtrees of $c(\R_2)$ are free. Then we can have an extension of $\R_1$ and $\R$ with all free lost subtrees, by just removing one or all subtrees extending over $e$. Hence $\omega(c(\R))=\omega(c(\R_1))=\omega(c(\R_2))=|\Delta(\R)|$. This means that we can assign at least $l+1$ losses to $e$ in $\R$ without completion cost change. This contradicts the fact that $k'<l+1$. Therefore $c(\R_2)$ has at least one non-free lost subtree.

Let us prove that there exists a chain of lost subtrees $v_1,\ldots, v_{m-1},v_m$ (in $\R_2$) such that $v_1< \ldots < v_m$, $v_i$ overlaps $v_{i+1}$, $(i=1,\ldots,m-1)$, $v_1$ is a non-free tree, $v_2,\ldots,v_m$ are free trees and $v_m$ is a tree assigned to $p(s)$ extending over $e$. 

Assume the opposite. Let $T_S$ be the maximum subtree with root edge $e$ that contains only free lost subtrees (see Figure \ref{fig:maximal_subtree}), and  $f_1,\ldots,f_r$ edges of $S$ that are children of leaf-edges of $T_S$. Because of the maximality of $T_S$ and the assumption that there is no chain leading from non-free tree to one of the trees $t_1,\ldots,t_{l+1}$, we have that there is no tree expanding from inner node of $T_S$ over one of the edges $f_1,\ldots,f_r$. Since $\R_2$ has at least one non-free lost subtree, we have $r\ge 1$, {\em i.e.} edges $f_1,\ldots,f_r$ do exist.

Since $\omega(c(\R))<\omega(c(\R_2))$ and $c(\R_2)$ has only free trees in $T_S$, then there is $i$ such that $\omega(c(\R)(f_i)) < \omega(c(\R_2)(f_i))$. Since no lost subtrees expands from inner node of $T_S$ over $f_i$, we can take the lost subtrees with roots in $c(\R)(f_i)$ and use them in $c(\R_2)$, instead  of the lost subtrees in $c(\R_2)(f_i)$. Thus we obtain an extension of $\R_2$ with strictly less cost than $c(\R_2)$, a contradiction. This means that there is a chain $v_1,\ldots,v_m$ with described properties ($v_1$ is non-free, {\em etc.}).

Now, apply switch operation on $v_i,v_{i+1}$, for every $i=1,\ldots,m-1$. In this way $v_m$, which is one of the trees $t_1,\ldots,t_{l+1}$, becomes non-free. The weight of $c(\R_2)$ is not changed with these switch operations. Now, by removing $v_m$, we obtain an extension of $\R_1$ with strictly less cost than $c(\R_2)$, which contradicts the assumption $\omega(c(\R_2)) = \omega(c(\R_1))$. Therefore $\omega(c(\R_2)) = \omega(c(\R_1))+1$.  \qed
\end{proof}

\begin{figure}[H]	
    \centering 
	\includegraphics{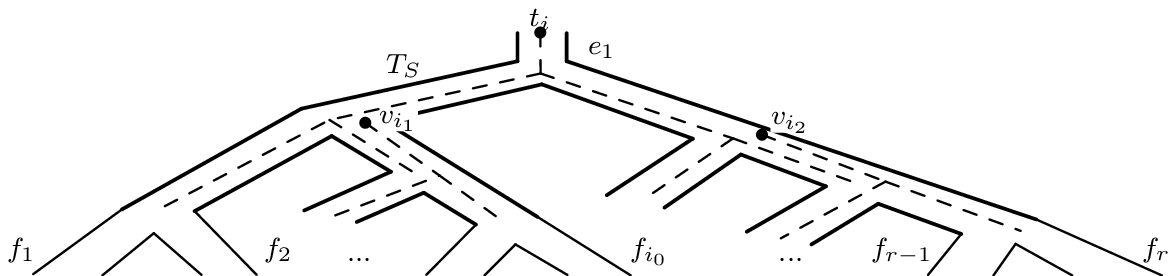}
	\caption{Tree $T_S$ is the maximum subtree of $S(e_1)$ rooted at $e_1$ that contains only free lost subtrees. There are no lost subtrees expanding from $T_S$ over $f_j$. Black dots denote roots of lost subtrees ($t_i,v_{i_1},v_{i_2}$). All non-free lost subtrees of $S(e_1)$ are in $S(f_j)$, $j=\overline{1,r}$}
	\label{fig:maximal_subtree}
\end{figure}

 Let $e=(s,p(s))\in E(S)$, and by $\L'(e)=\L'(e,\R)=\L'(s)=\L'(s,\R)$ denote the number of non-free lost subtrees, in the reconciliation $\R$, with a root assigned to $p(s)\in V(S)$,  expanding over $e$.
\begin{lemma}\label{lem:nonFreeLossesInCompletion}
	Let $\R$ be a reconciliation, $\R_1$ output of \textproc{OneCompletion}$(\R)$, $e\in E(S)$, and $e_l,e_r$ children of $e$. Then
	\begin{enumerate}[(a)]
		\item $F(e,\R)>0 \implies \L'(e,\R_1)=0$;
		\item $F(e,\R)\le 0 \implies \L'(e,\R_1)=\L(e,\R)-\D(e,\R)-max(min(F(e_l,\R),F(e_r,\R)),0)$;
		\item \label{lem:nonFreeLossesInCompletion-c} if  $\R_2$ is another output of \textproc{OneCompletion}$(\R)$, then $\omega(\R_1)=\omega(\R_2)$.  
	\end{enumerate}
\end{lemma}
\begin{proof}
	Let $l=\L(e,\R), d=\D(e,\R), m=max(min(F(e_l,\R),F(e_r,\R)),0)$.
	
	$(a)$ Since \textproc{OneCompletion} extends losses only into edges $e$, if $F(e)>0$, we have that the number of extra losses, expanded over $e$, is not greater than $F(e)$. Assume that $\R_1$ generates extra $f$ losses in $e$. Hence $f\le F(e,\R)=m+d-l$. Let $f_m$ and $l_m$ be the number of losses made free by extending over $e_l,e_r$, and $f_d$ and $l_d$ are number of losses made free by assigning them to the duplications in $e$. Hence $f_m+l_m\le m$, $f_d+l_d\le d$, $f_m+f_d\le f$, $l_m+l_d\le l$. 
	
	Assume the opposite,  let $\L'(e,\R_1)>0$. Then $f_m+f_d+l_m+l_d<f+l\le d+m$ $\implies$ $f_d+l_d < d$ or $f_m+l_m < m$. Therefore one extra loss can be made free by assigning it to duplication in $e$, or extending it over $e_r$ and $e_l$. This contradicts the procedure \textproc{ExtendLossIntoFreeTree}, which make loss free if $\Delta(e)\ne \emptyset$ or $F(e_1)>0$ and $F(e_2)>0$. 
	
	$(b)$  Since $F(e,\R)\le 0$, $\R_1$ does not extend any new losses over $e$, and $m+d\le l$. At most $m$ losses can be extended over $e_r$ and $e_l$, and at most $d$ losses can be assigned to the duplications in $e$. Therefore, number of losses that remained non-free is $l-d-m$.

	$(c)$ From $(a)$ and $(b)$, we have $\L'(e,\R_1)=\L'(e,\R_2)$, $\forall e\in E(S)$, hence $|\Lambda \backslash \Lambda'(\R_1)|=|\Lambda \backslash \Lambda'(\R_2)|$. Since \textproc{ExtendLossIntoFreeTree} does not create new duplications, we have $\Delta(\R_1)=\Delta(\R_2)=\Delta(\R)$. Therefore $\omega(\R_1)=\omega(\R_2)$. \qed
\end{proof}

\begin{lemma}\label{lem:Completion}	
	Let $\R$ be a minimal reconciliation. Then  \textproc{AllCompletions}$(\R)$ returns a completion of $\R$.
\end{lemma}
\begin{proof}
	Let $\R_1$ be a reconciliation from Lemma \ref{lem:arrangedLostSubtrees}, obtained by applying switch operations on $c(\R)$. Then $\omega(\R_1)=\omega(c(\R))$ and $\R_1$ satisfies the conditions from Lemma \ref{lem:CompletionOutput}. Hence $\R_1$ is a possible output of the series of procedures \textproc{OneCompletion}$(\R)$, \textproc{ExtendLossesIntoNonFreeTrees}$(\R)$. 
	
	Let $\R_2$ be another output of this series of procedures with the input $\R$. From Lemmas \ref{lem:ExtendLossIntoFreeTree} and \ref{lem:nonFreeTree} we have that $\R_2$ is an extension of $\R$. From Lemmas \ref{lem:nonFreeLossesInCompletion} (\ref{lem:nonFreeLossesInCompletion-c}) and \ref{lem:ExtendNonFreeLosses} we have $\omega(\R_1)=\omega(\R_2)$. Since $\R_1$ is a completion of $\R$, we have  $\R_2$ is a completion of $\R$.   
	
	Since \textproc{Switch} does not change the weight of a reconciliation (Lemma \ref{lem:switchIsConst}) and $\R_2$ is a completion of $\R$, we have that \textproc{AllCompletions}$(\R)$ is also a completion of $\R$. \qed
\end{proof}

\begin{theorem}\label{th:theorem_1}
 	Algorithm $\ref{randR}$ returns an optimal solution.
\end{theorem}
\begin{proof}
 The algorithm starts with LCA reconciliation $\R_1$. LCA's completion is an optimal reconciliation (Theorem \ref{th:LCAopt}), therefore completion of $\R_1$ is an optimal reconciliation. 
 
 Let $\R_2$ be an output of \textproc{RaiseSeveralDuplications}$(\R_1)$. Then $c(\R_2)$ is an optimal reconciliation (Lemma \ref{lem:Changeposition}).   
 
 Let $\R_3$ be an output of \textproc{AllCompletions}$(\R_2)$. Then (Lemma \ref{lem:Completion}) it is s completion of $\R_2$, hence $\R_3$ is an optimal reconciliation. 
 
 Assume that $\R_4$ is an output of \textproc{RaiseConversions}$(\R_3)$. From Lemma \ref{lem:RaiseConversions} we have $\omega(\R_4)=\omega(\R_3)$. Hence $\R_4$ is an optimal reconciliation. Note that $\R_4$ is an output of \textproc{RandR}$(S,G,\phi)$. \qed 
\end{proof}

Next lemma states that all duplications raised on a path going though a vertex with non positive flow on its children are conversions.

\begin{lemma}\label{lem:Completion_2}
	Let $\R$ be a ZF reconciliation  such that
        if $v',v$ are non-free and free lost subtrees that overlap, then $v\le v'$.
	Then $\R$ is a possible output of \textproc{ExtendLossesIntoNonFreeTrees}.
\end{lemma}
\begin{proof}
  From Lemma \ref{lem:RaiseDuplication} we have that $\R'$ is a possible output of \textproc{RaiseDuplication}, where $\R'$ is the minimization of $\R$. From Lemma \ref{lem:CompletionOutput} and and this Lemma condition, $\R$ is a possible output of the series of procedures $\textproc{OneCompletion}(\R')$, $\textproc{ExtendLossesIntoNonFreeTrees}(\R')$. Hence $\R$ is a possible output of \textproc{ExtendLossesIntoNonFreeTrees}.  \qed   
\end{proof}

\begin{lemma}\label{lem:SwitchOutput}
	Let $\R$ be a ZF reconciliation. Then $\R$ is a possible output of \textproc{Switch}. 
\end{lemma}
\begin{proof}
	Let $v'$ and $v$ be non-free and free lost subtrees in $\R$. If they overlap and $v'<v$, apply switch operation. Previous procedure repeat as long as there are such trees. Let us prove that the procedure will stop.
	
	Let $d$ be the sum of the distances of the roots of the non-free subtrees to $root(S)$. With every switch operation $d$ decreases. Since $d\ge 0$, it cannot decrease indefinitely. Hence the procedure will stop.
	
	The reconciliation, obtained in this way, denote by $\R_1$. Now, $\R_1$ satisfies the conditions in Lemma \ref{lem:Completion_2}, hence it is a possible output of \textproc{ExtendLossesIntoNonFreeTrees}. 
	
	So, by \textproc{ExtendLossesIntoNonFreeTrees} we obtain $\R_1$, and by $\textproc{Switch}(\R_1)$, where switch operations are applied in the reversed order, we obtain $\R$. \qed           
 \end{proof}

\begin{theorem}\label{th:allGenerated}
	Any optimal solution can be generated by Algorithm \ref{randR}.
\end{theorem}
\begin{proof}
  Let $\R$ be an arbitrary optimal reconciliation.	By lowering some conversions, we can obtain a ZF reconciliation $\R_1$ such that $\omega(\R_1)=\omega(\R)$ (see Lemma \ref{lem:L>0_X=Y=0}).
  
  By Lemma \ref{lem:SwitchOutput},  $\R_1$ is obtainable by \textproc{Switch}. 
  
  So, $\R_1$ is a possible output of \textproc{Switch}, and $\R$ is a possible output of $\textproc{RaiseConversions}(\R_1)$, if conversion raising is applied in the reversed order.           \qed
\end{proof}

\begin{theorem}
 Algorithm \ref{randR} has time complexity $O(m^2+m\cdot n)$.
\end{theorem}
\begin{proof}
	  Let $n=|V(G)|$, $m=|V(S)|$, then $E(G)\in O(n), E(S)\in O(m)$. LCA reconciliation can be determined in linear time (see \cite{Chauve2009}), say $O(m+n)$. 
	  
	  Algorithm \ref{ExtendLoss} forms a set $\Delta''(e)$ and it takes $O(m)$ time. It extends a loss into free tree. The maximum size of a (non-)free tree is $O(m)$. Algorithm \ref{ExtendFreeLosses} applies Algorithm \ref{ExtendLoss} $|\Sigma\backslash \Sigma'|\le |\Sigma|$ times, hence it has time complexity $O(|\Sigma|\cdot m)$.
	  
	  Algorithm \ref{PossiblePositions}  determines possible new positions for a duplication $d$. Since the height of the tree $S$ is $O(m)$, we have that the number of possible positions is also $O(m)$ and this is the complexity of Algorithm \ref{PossiblePositions}. Algorithm \ref{dPosition} calls Algorithm \ref{PossiblePositions} and generates $k\in O(m)$ new losses. Hence the complexity of Algorithm \ref{dPosition} is $O(m)$. Algorithm \ref{RaiseSeveralDuplications} calls Algorithm \ref{dPosition} $|\Delta|$ times and its complexity is $O(|\Delta|\cdot m)$.
	  
	  Algorithm \ref{raise_one} raises one conversion. Maximal raise height is $O(m)$ and this is the complexity of the algorithm. Algorithm \ref{raise_some} calls Algorithm \ref{raise_one} $|C|$ times ($C$ - the set of all conversions). Therefore the complexity of Algorithm \ref{raise_some} is $O(|C|\cdot m).$   
	  
	  Algorithm \ref{ExtendOneNonFreeLoss} extends a loss into a non-free tree. The size of non-free tree is $O(m)$ an this is the complexity of the algorithm. Algorithm \ref{ExtendNonFreeLosses} uses Algorithm \ref{ExtendOneNonFreeLoss} $|\Sigma_1|$ times, and its complexity is $O(|\Sigma_1|\cdot m)$. 
	  
	  Algorithm \ref{switch} applies a switch operation on lost subtrees. With every switch, a root of a subtree with non-free loss is further away from $root(S)$. Longest distance from $root(S)$ is $O(m)$. Switch operation always include one non-free loss. Therefore the complexity of this algorithm if $O(|\Sigma\backslash\Sigma'|\cdot m)$.
	  
	  When we add corresponding complexities we get $O(m+n) + O(|\Sigma|\cdot m) + O(|\Delta|\cdot m) + O(|C|\cdot m) + O(|\Sigma_1|\cdot m) + O(|\Sigma\backslash\Sigma'|\cdot m)$. Since $|\Sigma|, |\Sigma_1|, |\Sigma\backslash\Sigma'| \in O(m+n)$, $|\Delta|\in O(n)$, we have that the complexity of the main algorithm is $O(m^2+m\cdot n)$.    \qed
\end{proof}

	\section{Conclusion}

In this paper we give a polynomial algorithm that returns an optimal reconciliation in duplication, loss, conversion model. The algorithm can return any optimal reconciliation with a non-zero probability, and can enumerate the whole space of solutions.

A natural extension would be a uniform sampling of all solutions in order to statistically assess properties of the solution space. Because of the switch operation, this could be achieved by an Markov chain Monte Carlo method. Future work is to define adequate transition probabilities to ascertain fast convergence.

An interesting problem that we leave open for further research is the weighted case. Unfortunately the approach, used in this paper, is not useful for this case. A completion of LCA reconciliation does not have to be an optimal reconciliation (see Figure \ref{fig:weighted_case}). It might be necessary to raise some speciations from $V(G)$ in order to obtain an optimal solution. 

Adding transfers and recombinations significantly increases the complexity of the problem.

\begin{figure}[!h]
    \centering 
    \includegraphics{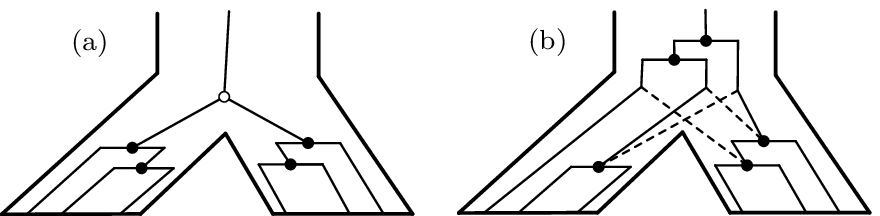}
	\caption{Weighted case and $(d,l,c)=(2,1,1)$. (a) LCA reconciliation is equal to its completion (because there are no losses), and the weight is $4d=8$. (b) The speciation and duplication are raised. Speciation is now duplication and three new losses are added. The weight is $2d+3c=7$}
    \label{fig:weighted_case}
\end{figure}



\bibliographystyle{spbasic}      


\bibliography{bibliography} 
\end{document}